\documentclass[12pt]{article}
\usepackage{amsmath,centernot}
\usepackage{amsfonts}
\usepackage{mathtools}
\usepackage{amsthm}
\usepackage{graphicx}
\usepackage{enumerate}
\usepackage{natbib}
\usepackage{url} 
\usepackage{authblk}

\usepackage{float}
\usepackage{subfigure}
\usepackage{enumitem}

\usepackage{tikz}
\usetikzlibrary{arrows,automata,positioning}
\usepackage{graphicx}
\tikzstyle{ov}=[shape=rectangle,
                draw=black!50,
                thick,
                minimum width=0.7cm,
                minimum heighj=0.7cm]

\tikzstyle{av}=[shape=rectangle,
                draw=black!50,
                fill=black!10,
                thick,
                minimum width=0.7cm,
                minimum heighj=0.7cm]

\tikzstyle{lv}=[shape=circle,draw=black!50,thick]

\usetikzlibrary{arrows,positioning}
\tikzset{
   >=stealth',
   punkt/.style={rectangle, rounded corners, draw=black, very thick, text width=10cm, minimum height=1cm, text centered, color=blue!50!black},
       punkt2/.style={rectangle, rounded corners, draw=black, very thick, text width=7cm, minimum heighj=0.7cm, text centered, color=blue!50!black},
   pil/.style={->, thick, color=blue!50!black}
}

\newcommand{\indep}{\rotatebox[origin=c]{90}{$\models$}}

\newcommand{\nindep}{\centernot{\indep}}

\newtheorem{theorem}{Theorem}[section]
\newtheorem{corollary}{Corollary}[theorem]
\newtheorem{proposition}{Proposition}[theorem]

\newtheorem{assumption}{Assumption}
\newtheorem{definition}{Definition}
\newtheorem{example}{Example}

\newtheorem*{remark}{Remark}

\usepackage{setspace}
 \usepackage[vmargin=1.18in,hmargin=1.1in]{geometry}

\begin{document}

\title{Using negative controls to identify causal effects\\ with invalid instrumental variables}

\author[1]{Oliver Dukes}
\author[2]{David B. Richardson}
\author[3]{Zachary Shahn}
\author[4]{James M. Robins}
\author[5]{Eric J. Tchetgen Tchetgen}

\affil[1]{Ghent University, Ghent,  Belgium}
\affil[2]{University of California Irvine,  Irvine,  CA,  USA}
\affil[3]{CUNY School of Public Health,  New York,  NY,  USA}
\affil[4]{Harvard T. H. Chan School of Public Health,  Boston,  MA,  USA}
\affil[5]{University of Pennsylvania, Philadelphia, PA,  USA}

\date{}

\maketitle

\begin{abstract}
Many proposals for the identification of causal effects require an instrumental variable that satisfies strong, untestable unconfoundedness and exclusion restriction assumptions. In this paper, we show how one can potentially identify causal effects under violations of these assumptions by harnessing a negative control population or outcome. This strategy allows one to leverage sub-populations for whom the exposure is degenerate, and requires that the instrument-outcome association satisfies a certain parallel trend condition. We develop the semiparametric efficiency theory for a general instrumental variable model, and obtain a multiply robust, locally efficient estimator of the average treatment effect in the treated. The utility of the estimators is demonstrated in simulation studies and an analysis of the Life Span Study.
\end{abstract}


\doublespacing

\section{Introduction}

There is now a long tradition in causal inference of developing identification strategies that do not rely on a `no unmeasured confounding' (or conditional exchangeability) type assumption.  One of the most popular of these is the `instrumental variable' approach. For validity, instrumental variable-based causal inference relies on having access to a variable that satisfies the three core instrumental variable conditions: IV.1 it is associated with the exposure (relevance); IV.2 it affects the outcome only via the exposure (exclusion restriction); and IV.3 it shares no unmeasured common causes with the outcome (unconfoundedness) \citep{hernan2006instruments}.  

Unfortunately, in much applied work the three core instrumental variable conditions may be violated. Recently, \citet{davies2017compare} and \citet{danieli2022negative} have considered how one can assess violations of IV.2 and IV.3 using \textit{negative controls} \citep{lipsitch2010negative,sofer2016negative}. For example, physician preference is a popular choice of instrument in pharmacoepidemiologic studies, where interest might be on the causal effect of a particular treatment for a given disease. A negative control \textit{population} for the instrument could be constructed based on patients for whom it is known that the treatment has no effect. Under IV.2 and IV.3, in this population physician preference should not be associated with the outcome of interest at all; hence any association is indicative of a violation of the exclusion restriction and/or unconfoundedness. In this work, we will introduce a more general concept of a \textit{reference population}, of which a negative control population is a special case.

Alternatively, a negative control \textit{outcome} is a variable that cannot plausibly be caused by the exposure. Under the core instrumental variable assumptions, the instrument would then not be associated with the negative control outcome. For example, one might choose a clinical outcome unlikely to be affected by the medication. There is an existing literature on using negative control outcomes as a means to detect and remove bias in an analysis based on confounding adjustment \citep{lipsitch2010negative,sofer2016negative}; however, they have been used less for assessing the instrumental variable assumptions. \citet{davies2017compare} describe how negative control outcomes can be used to compare the plausibility of instrumental variable approaches with analyses that assume no unmeasured confounding; \citet{sanderson2021use} use these variables to check for violations of unconfoundedness due to population stratification in Mendelian Randomisation studies.

We go beyond falsification testing and propose conditions for identification and estimation of effects with an invalid instrument variable. Our results rely on a condition on the stability of the association between the instrument and outcome, which is closely related to the \textit{parallel trends} assumption in the difference-in-difference literature. Intuitively, our assumption states that the instrument-outcome association that arises due to violations of IV.2 and IV.3 in the na\"ive instrumental variable analysis must equal the instrument-outcome association in a reference population. If our assumption holds, the potential instrument does not necessarily have to satisfy the conditions IV.2 and IV.3, but still must be predictive of the exposure, therefore satisfying IV.1. A special feature of our framework, relative to other negative control strategies, is that it allows one to leverage sub-populations who were \textit{not eligible} to be exposed. We hope that our methodology may be generally useful in epidemiologic applications where such sub-populations arise. 

In this paper, we first formalise conditions for identification of the average treatment effect on the treated given a potentially invalid instrument and negative control/reference population. These results hold interesting connections with recent work on instrumental variables, difference-in-differences, and data fusion. We then develop a theory of estimation and inference for these causal estimands, working under a semiparametric model for the observed data distribution. This is in contrast to the existing work on falsification testing, which is developed in the context of linear regression models and two stage least squares estimation. Our theory yields several novel classes of estimators, which are distinct from those proposed elsewhere in the instrumental variable literature. Our estimators turn out to have a \textit{quadruple robustness} property. Specifically, if the semiparametric model for the causal contrast is correct, then our estimators are unbiased if one out of four sets of additional constraints on the observed data distribution hold. The methods are illustrated in simulation studies and an analysis of the Life Span Study. Our results are closely related to and develop from the \textit{bespoke instrumental variable} design recently proposed in occupational epidemiology \citep{richardson2021bespoke}. That work was not developed with specific reference to leveraging negative controls in instrumental variable designs, and lacked general semiparametric estimation and efficiency theory, which we hereby provide. We also give new identification results. This theory may be of independent interest for those working on semiparametric instrumental variable models, as we leverage ideas from unrelated work on generalised odds ratio modelling to construct quadruply robust estimators which to our knowledge has not been done before.

\section{Identification}

\subsection{Bias detection using reference populations}

We will first consider a setting where we have i.i.d. copies of a point-treatment exposure $A$ (for the moment assumed to be binary), an end of study outcome $Y$, a vector of covariates $\mathbf{C}$, and a potentially invalid (and, for the moment, binary) instrument $Z$. All are measured in a population of interest ($S=1$).
Suppose furthermore that data are also available on an additional population ($S=0$) satisfying a key property given in Assumption 1 below and for which $Y$ and $Z$ are observed. Let $Y^a$ define the potential outcome that would have been observed, had an individual been assigned to intervention level $A=a$.

In what follows, the population for whom $S=0$ will be used to correct for bias that may occur when $Z$ is an invalid instrumental variable. In order to do this, we will first list the assumptions required to detect potential bias, and eventually to identify the causal effect. 
\begin{assumption}Reference population: Let $\mathcal{C}_0$ ($\mathcal{C}_1$) denote the support of $\mathbf{C}$ in individuals with $S=0$ ($S=1$). Then $\forall z\in \{0,1\}$ and $\forall \mathbf{c} \in \mathcal{C}_0\cap \mathcal{C}_1$,\label{ref}
\[E(Y|Z=z,S=0,\mathbf{C}=\mathbf{c})=E(Y^0|Z=z,S=0,\mathbf{C}=\mathbf{c}).\]
\end{assumption}

Assumption \ref{ref} states that in the reference population, the conditional expectation of the observed outcome (conditional on $Z$ and $\mathbf{C}$) is equal to the conditional expectation of the potential outcome under no treatment.  This can be plausible if by an exogenous intervention which may be a physical or temporal restriction, individuals were unable to receive treatment. In pharmacoepidemiologic studies, a reference population may be sourced from time periods when a new treatment was not yet available, or from patients not eligible to receive the treatment.  Alternatively, in environmental health studies, a natural reference population may be determined based on physical/spatial considerations that would have prevented access to the exposure. 

A negative control population where the exposure is not necessarily degenerate but is hypothesised to have no effect, could also form a reference population. However, the concept of a reference population allows for the possibility that the treatment causally impacts the outcome in individuals with $S=0$, if these individuals were not excluded from access to treatment. In general, it is advantageous to choose a reference population that is similar in important baseline characteristics to the $S=1$ population. This is because to eventually identify a causal effect, we will need to transport associations across from the reference to the $S=1$ population. We emphasise that it is not obvious how other negative control-type designs can generally leverage sub-populations for whom treatment is degenerate. These designs usually involve contrasts between levels of $A$ that are expected to be zero when there is no bias. Hence some variation in exposure is typically required \citep{lipsitch2010negative}.

\begin{figure}
    \centering
    \begin{minipage}{.5\textwidth}
	\begin{tikzpicture}[scale=0.7]
\node[] (z) at (-2, 0)   {$Z$};
\node[] (a) at (0, 0)   {$A$};
\node[] (c) at (2, 2)   {$C$};
\node[] (y) at (4, 0)   {$Y$};
\node[draw,circle] (u) at (2,-2)   {$U$};
\node[] (r) at (0,-4)   {\textrm{(a) Population of interest ($S=1$)}};
\path[->] (z) edge node {} (a);
\path[->] (a) edge node {} (y);
\path[->] (c) edge node {} (a);
\path[->] (c) edge node {} (y);
\path[->] (c) edge node {} (z);
\path[->] (z) edge  [bend right] node {} (y);
\path[->,dashed] (u) edge node {} (y);
\path[->,dashed] (u) edge node {} (a);
\path[->,dashed] (u) edge node {} (c);
\path[->,dashed] (u) edge node {} (z);
\end{tikzpicture}
\end{minipage}%
\begin{minipage}{.5\textwidth}
	\begin{tikzpicture}[scale=0.7]
\node[] (z) at (-2, 0)   {$Z$};
\node[] (c) at (2, 2)   {$C$};
\node[] (y) at (4, 0)   {$Y$};
\node[draw,circle] (u) at (2,-2)   {$U$};
\node[] (r) at (0,-4)   {\textrm{(b) Reference population ($S=0$)}};
\path[->] (c) edge node {} (y);
\path[->] (c) edge node {} (z);
\path[->] (z) edge  [bend right] node {} (y);
\path[->,dashed] (u) edge node {} (y);
\path[->,dashed] (u) edge node {} (c);
\path[->,dashed] (u) edge node {} (z);
\end{tikzpicture}
\end{minipage}%
\caption{Possible DAG with a reference population}
\label{DAG}
\end{figure}

A potential data-generating process is illustrated in the causal diagram in Figure \ref{DAG}. Here, the edges from $Z$ to $Y$ and  $U$ to $Z$ in both populations indicate violations of IV.2 and IV.3. As noted in \citet{davies2017compare}, under Assumption \ref{ref} and conditions IV.2 and IV.3 holding in both populations, one would expect a null association between $Z$ and $Y$ in the reference population. This can be formalised as follows:
\begin{proposition}\label{main_prop}
Suppose that $\forall z \in \{0,1\}$ and $\forall \mathbf{c} \in \mathcal{C}_0 \cap \mathcal{C}_1$,
\begin{itemize}
    \item If $Z=z$ and $S=1$, then $Y^{0,z}=Y$ (instrument consistency).
    \item $E(Y^{0,z}|Z=z,S=0,\mathbf{C}=\mathbf{c})=E(Y^{0,z}|S=0,\mathbf{C}=\mathbf{c})$ (unconfoundedness). \item $E(Y^{0,z}|S=0,\mathbf{C}=\mathbf{c})=E(Y^{0}|S=0,\mathbf{C}=\mathbf{c})$ (exclusion restriction).
    \end{itemize}
Then under Assumption \ref{ref}, $\forall \mathbf{c} \in \mathcal{C}_0 \cap \mathcal{C}_1$
\begin{align}\label{nc_equal}
E(Y|Z=1,S=0,\mathbf{C}=\mathbf{c})=E(Y|Z=0,S=0,\mathbf{C}=\mathbf{c}).
\end{align}
\end{proposition}
See \ref{sm:proof_prop} of the appendix for a proof. This statement provides the basis of a valid test of IV.2 and IV.3 in the reference population. To draw conclusions about the population of interest, we would furthermore need to assume that
$\forall z \in \{0,1\}$ and $\forall \mathbf{c}\in \mathcal{C}_0 \cap \mathcal{C}_1$,
\begin{align}\label{shared_bias}
&E(Y^{0,z}|Z=z,S=1,\mathbf{C}=\mathbf{c})=E(Y^{0}|S=1,\mathbf{C}=\mathbf{c})\nonumber \\& \Rightarrow E(Y^{0,z}|Z=z,S=0,\mathbf{C}=\mathbf{c})=E(Y^{0}|S=0,\mathbf{C}=\mathbf{c}). 
\end{align}
In other words, $Z$ is an instrument in $S=0$ if it is an instrument in $S=1$. Then \eqref{nc_equal} can be used as the basis of a valid falsification test for the IV assumptions in the population of interest. If we strengthen \eqref{shared_bias} to an `iff' statement and Proposition \ref{main_prop} to an `iff' result (under a faithfulness-type condition), then we can construct a consistent test. We note that the proposition allows for the distribution of $\mathbf{C}$ to differ between populations, but if the support of $\mathbf{C}$ differs between the standard and reference populations, one may not be able to draw conclusions for all individuals with $S=1$.

\subsection{Point identification}

We will focus on identification of the conditional average treatment effect on the treated in the population of interest:
\[E(Y^1-Y^0|A=1, Z=z, S=1,\mathbf{C}=\mathbf{c}).\]
We focus on conditional effects given that investigators may often be interested in subgroup-specific contrasts. If one is interested in the marginal treatment effect in the treated, our identification and estimation results are still useful since conditional effects can be standardized according to the covariate distributions in the treated arm:
\begin{align*}
E(Y^1-Y^0|A=1,S=1)=E\{E(Y^1-Y^0|A=1, Z, S=1,\mathbf{C})|A=1, S=1\}.
\end{align*}
Further assumptions are required to identify a causal effect. 

\begin{assumption}{Consistency:}\label{consist}
if $A=a$ and $S=1$, then $Y^a=Y$. 
\end{assumption}

\begin{assumption}{IV relevance:}\label{relevance} $\forall \mathbf{c}\in \mathcal{C}_0 \cap \mathcal{C}_1$,
\[E(A|Z=1,S=1,\mathbf{C}=\mathbf{c})-E(A|Z=0,S=1,\mathbf{C}=\mathbf{c})\neq 0.\]
\end{assumption}
\begin{assumption}{Partial population exchangeability:}\label{partial_pe} $\forall \mathbf{c}\in \mathcal{C}_0 \cap \mathcal{C}_1$,
\begin{align*}
&E(Y^0|Z=1,S=1,\mathbf{C}=\mathbf{c})-E(Y^0|Z=0,S=1,\mathbf{C}=\mathbf{c})\\
&=E(Y^0|Z=1,S=0,\mathbf{C}=\mathbf{c})-E(Y^0|Z=0,S=0,\mathbf{C}=\mathbf{c}).
\end{align*}
\end{assumption}

We note first that Assumption \ref{consist} is standard in causal inference.  Assumption \ref{relevance} demands that $Z$ is predictive of the exposure; although this assumption is similar to IV.1, $Z$ is not required to satisfy either IV.2 or IV.3. In particular, our assumptions also allow for $Z$ to be a standard confounder with a direct effect on the outcome (violating the exclusion restriction). When IV.2 and IV.3 hold \textit{in both populations}, we note that Assumption \ref{partial_pe} then automatically also holds. This is because one can show that under the IV conditions, $E(Y^0|Z=1,S=s,\mathbf{C}=\mathbf{c})-E(Y^0|Z=0,S=s,\mathbf{C}=\mathbf{c})=0$ $\forall \mathbf{c}\in \mathcal{C}_0 \cap \mathcal{C}_1$ and $s=0,1$. Assumption \ref{partial_pe} requires that the (conditional) additive association between the potential outcome $Y_0$ and $Z$ transports between the negative control population and the and population of interest. It is not generally testable. 
We discuss the plausibility of this assumption below, which is weaker than the full population exchangeability assumption made in the transportability literature \citep{dahabreh2019generalizing}.

\begin{remark}
To provide some intuition for Assumption \ref{partial_pe}, in \ref{sm:lsem} of the appendix we consider an example data-generating mechanism based on linear structural models. It indicates how although we are agnostic about whether core conditions IV.2 or IV.3 are violated, knowledge that IV.3 (unconfoundedness) holds could make Assumption \ref{partial_pe} more plausible. In that case, the association between $Z$ and $Y$ in the reference population could be interpreted as a direct causal effect; one may be on firmer grounds to transport a causal effect across populations.
\end{remark}

\begin{remark}
Note that Assumption \ref{partial_pe} is related to a more standard `additive equi-confounding'-type assumption:
 \begin{align} 
    &E(Y^0|A=1,Z=z,S=1,\mathbf{C}=\mathbf{c})-E(Y^0|A=1,Z=z,S=0,\mathbf{C}=\mathbf{c})\nonumber\\
    &=E(Y^0|A=0,Z=z,S=1,\mathbf{C}=\mathbf{c})-E(Y^0|A=0,Z=z,S=0,\mathbf{C}=\mathbf{c})\label{condition_pt}
    \end{align}
$\forall z \in \{0,1\}$ and $\forall \mathbf{c}\in \mathcal{C}_0 \cap \mathcal{C}_1$ \citep{sofer2016negative}, which contrasts $Y^0$ within levels of $A$. To interpret \eqref{condition_pt} in the standard difference-in-differences setting, $S$ plays the role of time and $Z$ is an arbitrary covariate. Assumption \ref{partial_pe} does not appear to be stronger nor weaker in general than \eqref{condition_pt}. Both can be motivated by assuming that the (conditional mean) association between an unmeasured confounder $U$ and $Y$ is constant across populations. We defer to \ref{sm:lsem} for a more precise discussion, but note that Assumption \ref{partial_pe} places restrictions on the (conditional mean) association between $Z$ and $Y_0$, with different restrictions on the conditional mean of $U$. It is not obvious how to leverage \eqref{condition_pt} when there is no variation in $A$ in the reference population (as in our data example). Furthermore, if many candidate $Z$s are collected, one can then choose between them to obtain identification. In contrast, \eqref{condition_pt} is restricted to holding for the treatment variable.
\end{remark}

\begin{remark}
Our assumptions are stated conditional on covariates $\mathbf{C}$. To ensure Assumption \ref{partial_pe} holds, we require that $\mathbf{C}$ includes all effect modifiers of the additive association between $Z$ and $Y^0$ that differ in distribution between populations. In terms of design implications, this suggests choosing a reference population that is similar to the population of interest in terms of characteristics that may modify the association between $Z$ and $Y^0$. As argued above, it is also advantageous to incorporate measured confounders for the effect of $Z$ on $Y$, as well as for the effect of $A$ on $Y$, since this may weaken the restrictions on $U$. 
\end{remark}

Assumptions \ref{ref}-\ref{partial_pe} will be sufficient to test the causal null hypothesis that $E(Y^1-Y^0|A=1,Z=z,S=1,\mathbf{C}=\mathbf{c})=0$, as formalised in \ref{sm:result_ht} of the appendix. However, they are insufficient to identify the causal effect of interest. Therefore we will posit two additional assumptions; only one is required to hold to yield identification: 
\begin{assumption}{No effect modification (NEM):} \label{pem} $\forall \mathbf{c}\in \mathcal{C}_0 \cap \mathcal{C}_1$
\[E(Y^1-Y^0|A=1,Z=1,S=1,\mathbf{C}=\mathbf{c})=E(Y^1-Y^0|A=1, Z=0,S=1,\mathbf{C}=\mathbf{c}).\]
\end{assumption}
\begin{assumption}{No selection modification (NSM):}\label{psb}
If we define the selection bias as \[\gamma(z,\mathbf{c}) \equiv E(Y^0|A=1,Z=z,S=1,\mathbf{C}=\mathbf{c})-E(Y^0|A=0,Z=z,S=1,\mathbf{C}=\mathbf{c})\]
 then $\gamma(1,\mathbf{c})=\gamma(0,\mathbf{c})$ $\forall \mathbf{c}\in \mathcal{C}_0 \cap \mathcal{C}_1$.
\end{assumption}
The first of these requires that the conditional average treatment effect on the treated does not depend on $Z$ \citep{hernan2006instruments},  whereas the second assumption requires that the selection bias is not a function of $Z$ \citep{tchetgen2013alternative}.  Then we are in a position to give the main theorem for identification. 
\begin{theorem}\label{iden}
Under Assumptions \ref{ref} and \ref{partial_pe}, $t(z,\mathbf{c})\equiv E(Y^0|Z=z,S=1,\mathbf{C}=\mathbf{c})-E(Y^0|Z=0,S=1,\mathbf{C}=\mathbf{c})$ is identified as 
\begin{align*}
t(z,\mathbf{c})
&=E(Y|Z=z,S=0,\mathbf{C}=\mathbf{c})-E(Y|Z=0,S=0,\mathbf{C}=\mathbf{c})
\end{align*}
$\forall z \in \{0,1\}$ and $\forall \mathbf{c}\in \mathcal{C}_0 \cap \mathcal{C}_1$. Further assuming Assumptions \ref{consist} and \ref{relevance}, if Assumption \ref{pem} holds, then
\begin{align*}
&E(Y^1-Y^0|A=1,S=1,\mathbf{C}=\mathbf{c})\\
&=\frac{E(Y|Z=1,S=1,\mathbf{C}=\mathbf{c})-E(Y|Z=0,S=1,\mathbf{C}=\mathbf{c})-t(1,\mathbf{c})}{E(A|Z=1,S=1,\mathbf{C}=\mathbf{c})-E(A|Z=0,S=1,\mathbf{C}=\mathbf{c})}.
\end{align*}
Alternatively, if Assumption \ref{psb} also holds, then
\begin{align*}
&E(Y^1-Y^0|A=1, Z=z,S=1,\mathbf{C}=\mathbf{c})\\
&=E(Y|A=1,Z=z,S=1,\mathbf{C}=\mathbf{c})-E(Y|A=0,Z=z,S=1,\mathbf{C}=\mathbf{c})\\
&\quad +\frac{E(Y|A=0,Z=1,S=1,\mathbf{C}=\mathbf{c})-E(Y|A=0,Z=0,S=1,\mathbf{C}=\mathbf{c})-t(1,\mathbf{c})}{E(A|Z=1,S=1,\mathbf{C}=\mathbf{c})-E(A|Z=0,S=1,\mathbf{C}=\mathbf{c})}.
\end{align*}
\end{theorem}
A proof is given in \ref{sm:proof_iden} of the appendix. Note that the identification functional under NEM amounts to the conditional Wald estimand \citep{wang2018bounded}, after transforming the outcome by subtracting $t(z,\mathbf{C})$. Essentially, the bias (learnt in the reference population) is subtracted from the numerator of the Wald estimand, in a similar fashion to what is done in difference-in-differences. 

\subsection{Negative control outcomes}

One can also harness measurements of a negative control outcome in order to generate a reference population. Suppose that there is now only data on individuals from the population of interest, such that $S=1$ for everyone (we omit this from the conditioning statement to simplify notation). Further, suppose that one collects a random variable $W$ in the population of interest that satisfies the following assumption 
\begin{align}\label{ref_difference-in-differences}
E(W|Z=z,\mathbf{C}=\mathbf{c})=E(W^0|Z=z,\mathbf{C}=\mathbf{c})
\end{align}
$\forall z\in \{0,1\}$ and $\forall \mathbf{c}$ in the support of $\mathbf{C}$ \citep{danieli2022negative}. Intuitively, the observed outcome for $W$ is the same outcome as would be observed under no exposure; this is plausible when scientific expertise suggests there cannot be any exposure effect of $A$ on $W$. For example, $W$ could be an outcome occurring before the individual could be exposed. Note that \eqref{ref_difference-in-differences} restricts the effect of $A$ on the negative control outcome; so far, nothing is assumed about $Z$ being a valid instrument w.r.t $W$.   

If we additionally assume condition that
\begin{align}\label{partial_pe_difference-in-differences}
&E(Y^0|Z=1,\mathbf{C}=\mathbf{c})-E(Y^0|Z=0,\mathbf{C}=\mathbf{c})\nonumber\\&=E(W^0|Z=1,\mathbf{C}=\mathbf{c})-E(W^0|Z=0,\mathbf{C}=\mathbf{c})
\end{align}
then it follows from Theorem \ref{iden} that one can identify the effect $E(Y^1-Y^0|A=1,Z=z,\mathbf{C}=\mathbf{c})$ by leveraging the (conditional) association between the negative control outcome $W$ and $Z$.  Specifically,  if we assume that the causal effect does not depend on $Z$, 
\begin{align}\label{bsiv_difference-in-differences1}
&E(Y^1-Y^0|A=1, Z=z,\mathbf{C}=\mathbf{c})=\frac{E(Y-W|Z=1,\mathbf{c})-E(Y-W|Z=0,\mathbf{c})}{E(A|Z=1,\mathbf{c})-E(A|Z=0,\mathbf{c})}.
\end{align}
In \ref{sm:nco_id} of the appendix, we provide a parallel result when restricting the dependence of the selection bias on $Z$. 

In the case that $W$ is a pre-exposure outcome, \citet{ye2020instrumented} arrive at an  expression similar to (\ref{bsiv_difference-in-differences1}), although they consider a setting where two measurements of the exposure (in addition to the outcome) are available. In the denominator of their expression, the contrast involves the difference in the exposure measurements. In addition,  \citet{ye2020instrumented} target a different causal estimand and their identification results are separate from ours; see also \citet{richardson2023generalized}. Their identification functional also appears in the econometrics literature on `fuzzy' difference-in-differences designs \citep{de2018fuzzy}, where it is usually interpreted under separate assumptions as a local average treatment effect. 

\section{Estimation and inference}\label{sec:estimation}

\subsection{Semiparametric Theory}

In this section,  we will first develop a novel semiparametric efficiency theory for the conditional treatment effect $E(Y^{\mathbf{a}}-Y^{\mathbf{0}}|\mathbf{A}=\mathbf{a},\mathbf{Z}=\mathbf{z},S=1,\mathbf{C}=\mathbf{c})=\beta(\mathbf{a,z,c})$. We allow for $\mathbf{A}$, $\mathbf{Z}$ and $\mathbf{C}$ to be continuous or discrete and potentially vector-valued. Our results are distinct from those developed previously in the instrumental variable literature \citep{robins1994correcting,tchetgen2013alternative}. We note that our model parametrizes treatment effects in different populations; $\beta(\mathbf{a,z,c})$ may not equal $\beta(\mathbf{a',z,c})$ if $\mathbf{a}\neq\mathbf{a'}$. See \citet{robins1994correcting} for details on the parametrisation of structural mean models.

We postulate a semiparametric model 
\begin{align}\label{mod_res}
\beta(\mathbf{a,z,c})=\beta(\mathbf{a,z,c};\psi),
\end{align} where  $\psi$ is an unknown finite-dimensional parameter vector with true value $\psi^\dagger$. Then the previous NEM assumption can be generalised as:
\begin{assumption}{No effect modification (NEM):}\label{pem_gen} with probability 1, 
\[\beta(\mathbf{A,Z,C})=\beta(\mathbf{A,C}).\]
\end{assumption} We will let $\mathcal{M}_{NEM}$ denote the model defined by restriction (\ref{mod_res}) along with Assumptions \ref{pem_gen} and \ref{ref_gen}-\ref{partial_pe_gen}, (see \ref{sm:add_assum} of the appendix), which are 
generalised versions of Assumptions \ref{ref}-\ref{partial_pe}. To simplify the exposition, we work under a NEM assumption in the main manuscript, but parallel results under an NSM assumption are developed in \ref{sm:nsm} of the appendix.

Introducing some notation, let us redefine $t(\mathbf{z,C})\equiv E(Y^{\mathbf{0}}|\mathbf{Z=z},S=1,\mathbf{C})-E(Y^{\mathbf{0}}|\mathbf{Z=0},S=1,\mathbf{C})$.
Then, we define the residual
\begin{align*}
\epsilon^*&\equiv Y-\beta(\mathbf{A,C})S-t(\mathbf{Z,C})-b_{1}(\mathbf{C})S-b_{0}(\mathbf{C})
\end{align*}
where $t(\mathbf{Z,C})$ is identified as
\begin{align*}
&t(\mathbf{z,C})=E(Y|\mathbf{Z=z},S=0,\mathbf{C})-E(Y|\mathbf{Z=0},S=0,\mathbf{C})
\end{align*}
under \ref{consist_gen}-\ref{relevance_gen}. Also,
\begin{align*}
&b_1(\mathbf{C})\equiv E\{Y-\beta(\mathbf{A,C})S|\mathbf{Z=0},S=1,\mathbf{C}\}-E(Y-\beta(\mathbf{A,C})S|\mathbf{Z=0},S=0,\mathbf{C})\\
&b_0(\mathbf{C})\equiv E(Y|\mathbf{Z=0},S=0,\mathbf{C}).
\end{align*}.
We use the notation $\epsilon^*(\psi^\dagger)$ when $\beta(\mathbf{A,C};\psi^\dagger)$ replaces $\beta(\mathbf{A,C})$ in $\epsilon^*$.

\begin{theorem}\label{if_space}
Under the semiparametric model $\mathcal{M}_{NEM}$, the orthocomplement to the nuisance tangent space is given as
\begin{align*}
\bigg\{& \phi(\mathbf{Z},S,\mathbf{C})\epsilon_i^*(\psi^\dagger):\phi(\mathbf{Z},S,\mathbf{C})\in\Omega\bigg\}\cap L^0_2
\end{align*}
where \[\Omega\equiv\{\phi(\mathbf{Z},S,\mathbf{C}):E\{\phi(\mathbf{Z},S,\mathbf{C})|\mathbf{Z,C}\}=E\{\phi(\mathbf{Z},S,\mathbf{C})|S,\mathbf{C}\}=0\}\]
and $L_2^0$ is the Hilbert space of zero mean, finite variance functions with dimension $dim(\psi^\dagger)$.
\end{theorem}
A proof is given in \ref{sm:proof_if} of the appendix. As reviewed e.g. in \citet{bickel1993efficient}, by deriving the orthocomplement of the nuisance tangent space, we obtain the class of influence functions of all regular and asymptotically linear (RAL) estimators of $\psi^\dagger$. In turn, knowing this class motivates the construction of RAL estimators e.g. by solving estimating equations based on a chosen influence function. Specifically, under  Assumption \ref{pem_gen}, we can estimate $\psi^\dagger$ as the solution to the equations 
\begin{align*}
0=\sum^n_{i=1}\phi(\mathbf{Z}_i,S_i,\mathbf{C}_i)\epsilon_i^*(\psi^\dagger)
\end{align*}
for a chosen $\phi(\mathbf{Z},S,\mathbf{C})$ that satisfies the above restrictions.

Converting the above result into an estimation strategy requires a choice of $\phi(\mathbf{Z},S,\mathbf{C})$.  
In order to give a closed-form representation of $\Omega$,  we first provide a definition.
\begin{definition}{Admissible Independence Density \citep{tchetgen2010doubly}}.
Consider three potentially vector-valued random variables $\mathbf{X}_1$, $\mathbf{X}_2$ and $\mathbf{X}_3$. 
Let $f^\ddagger(\mathbf{X}_1,\mathbf{X}_2|\mathbf{X}_3)=f^\ddagger(\mathbf{X}_1|\mathbf{X}_3)f^\ddagger(\mathbf{X}_2|\mathbf{X}_3)$ denote a fixed density that makes $\mathbf{X}_1$ and $\mathbf{X}_2$ conditionally independent given $\mathbf{X}_3$. Then $f^\ddagger(\mathbf{X}_1,\mathbf{X}_2|\mathbf{X}_3)$ is an admissible independence density if it is absolutely continuous with respect to the true joint law $f(\mathbf{X}_1,\mathbf{X}_2|\mathbf{X}_3)$ with probability 1.
\end{definition}

It follows from \citet{tchetgen2010doubly} that, for the admissible independence density $f^\ddagger(\mathbf{Z},S|\mathbf{C})$ we have
\begin{align*}
\Omega=&\bigg\{\frac{f^\ddagger(\mathbf{Z},S|\mathbf{C})}{f(\mathbf{Z},S|\mathbf{C})}[r_0(\mathbf{Z},S,\mathbf{C})-E^\ddagger\{r_0(\mathbf{Z},S,\mathbf{C})|\mathbf{Z,C}\}\nonumber\\&-E^\ddagger\{r_0(\mathbf{Z},S,\mathbf{C})|S,\mathbf{C}\}+E^\ddagger\{r_0(\mathbf{Z},S,\mathbf{C})|\mathbf{C}\}]:\textrm{$r_0(\mathbf{Z},S,\mathbf{C})$ unrestricted}\bigg\}
\end{align*}
 where $E^\ddagger(\cdot|\cdot)$ denotes a (conditional) expectation taken with respect to $f^\ddagger(\cdot|\cdot)$. To construct an estimator, one can choose any $f^\ddagger(\mathbf{Z},S|\mathbf{C})$ that satisfies the above definition. The optimal choice of $r_0(\mathbf{Z},S,\mathbf{C})$ for efficiency can nevertheless depend on the choice of $f^\ddagger(\mathbf{Z},S|\mathbf{C})$. But there is no requirement to select the true densities, as the following example shows.
\begin{example}\label{bin_bin}
Let $\mathbf{Z}$ be binary, and suppose one sets $f^\ddagger(Z=1|\mathbf{C})=f^\ddagger(S=1|\mathbf{C})=0.5$ and chooses $r_0(Z,S,\mathbf{C})=16(Z-0.5)(S-0.5)m(\mathbf{C})$. 
Then one can estimate $\psi$ via the unbiased estimating function
\begin{align}\label{binary_ee}
\frac{m(\mathbf{C})(-1)^{Z+S}}{f(Z,S|\mathbf{C})}\epsilon^*(\psi^\dagger)
\end{align}
Equation \eqref{binary_ee} suggests a simple and practical choice of estimating function for binary $Z$, which we use in our simulations and data analysis.
\end{example}

\subsection{De-biased machine learning and multiply robust estimation}\label{sec_mr}

The estimation strategies previously described are not generally feasible, because the estimating equations for $\psi^\dagger$ involve nuisance parameters that are typically unknown. One approach to deal with this is to plug in estimates obtained from nonparametric estimators or flexible statistical learning methods \citep{chernozhukov2018double}. In \ref{sm:cross_fit} of the appendix, we therefore use Theorem \ref{if_space} to construct cross-fit de-biased machine learning-based estimators of $\psi^\dagger$, and describe sandwich estimators of the asymptotic variance. Nevertheless, to develop a more nuanced understanding of the different bias properties of RAL estimators of $\psi^\dagger$, we will instead consider a scenario where parametric working models are used for the conditional expectations/densities that arise in estimating the causal effect of interest. We hence conceptualise bias as potentially arising from parametric model misspecification. The parametric approach will also suggest how nuisance functions should be estimated nonparametrically for optimal performance in terms of their convergence rates.

Let us define the following parametric models $t(\mathbf{Z,C};\nu^{\dagger})$ for $t(\mathbf{Z,C})$, $b_0(\mathbf{C};\theta^{\dagger}_0)$ for $b_0(\mathbf{C})$ and $b_1(\mathbf{C};\theta^{\dagger}_1)$ for $b_1(\mathbf{C})$ that are required to model the conditional outcome mean for both the $S=1$ and $S=0$ populations. Here, $t(\mathbf{Z,C};\nu^{\dagger})$, $b_0(\mathbf{C};\theta^{\dagger}_0)$ and $b_1(\mathbf{C};\theta^{\dagger}_1)$ are known functions, which are smooth in $\nu^{\dagger}$, $\theta^{\dagger}_0$ and $\theta^{\dagger}_1$ respectively. We will sometimes use the notation $\theta^\dagger=(\theta^{\dagger^T}_1,\theta^{\dagger^T}_0)^T$. 
For modelling the joint conditional density $f(\mathbf{Z},S|\mathbf{C})$, we will make use of the following parametrisation based on the generalised odds ratio function 
\[OR(\mathbf{Z},S,\mathbf{C})=\frac{f(\mathbf{Z}|S,\mathbf{C})f(\mathbf{Z}=\mathbf{z_0}|S={s_0},\mathbf{C})}{f(\mathbf{Z}=\mathbf{z_0}|S,\mathbf{C})f(\mathbf{Z}|S={s_0},\mathbf{C})}\]
\citep{tchetgen2010doubly}. The reference values of $\mathbf{z_0}$ and $s_0$ are user-specified; in what fallows, we will use $\mathbf{z_0}=\mathbf{0}$ and $s_0=0$ as a generic notation. Then by specifying $OR(\mathbf{Z},S,\mathbf{C})$, $f(\mathbf{Z}|S=0,\mathbf{C})$ and $f(S|\mathbf{Z}=\mathbf{0},\mathbf{C})$, one can generate $f(\mathbf{Z},S|\mathbf{C})$ as 
\[f(\mathbf{Z},S|\mathbf{C})=\frac{OR(\mathbf{Z},S,\mathbf{C})f(\mathbf{Z}|S=0,\mathbf{C})f(S|\mathbf{Z}=\mathbf{0},\mathbf{C})}{\int \int OR(\mathbf{z},s,\mathbf{C})f(\mathbf{z}|S=0,\mathbf{C})f(s|\mathbf{Z}=\mathbf{0},\mathbf{C})\mathbf{dz}ds}\]
\citep{yun2007semiparametric} as well as $f(\mathbf{Z}|S,\mathbf{C})$ and $f(S|\mathbf{Z},\mathbf{C})$.
This parametrisation of the joint density will be crucial in constructing estimators with superior robustness properties. We postulate smooth parametric models $f(\mathbf{Z}|S=0,\mathbf{C};\tau^\dagger)$, $f(S|\mathbf{Z}=\mathbf{0},\mathbf{C};\alpha^\dagger)$ and $OR(\mathbf{Z},S,\mathbf{C};\rho)$ for $f(\mathbf{Z}|S=0,\mathbf{C})$, $f(S|\mathbf{Z}=\mathbf{0},\mathbf{C})$ and $OR(\mathbf{Z},S,\mathbf{C})$ respectively. We note that $\nu^\dagger$, $\theta^\dagger$, $\tau^\dagger$, $\alpha^\dagger$ and $\rho^\dagger$ are all finite dimensional parameters.

Consider the following sets of restrictions on the observed data distribution:
\begin{itemize}
\item $\mathcal{M}_1$: $t(\mathbf{Z,C})=t(\mathbf{Z,C};\nu^{\dagger})$, $b_0(\mathbf{C})=b_0(\mathbf{C};\theta^{\dagger}_0)$ and $b_1(\mathbf{C})=b_1(\mathbf{C};\theta^{\dagger}_1)$. 
\item $\mathcal{M}_2$: $f(\mathbf{Z}|S=0,\mathbf{C})=f(\mathbf{Z}|S=0,\mathbf{C};\tau^\dagger)$, $OR(\mathbf{Z},S,\mathbf{C})=OR(\mathbf{Z},S,\mathbf{C};\rho^\dagger)$\\ and $t(\mathbf{Z,C})=t(\mathbf{Z,C};\nu^{\dagger})$.
\item $\mathcal{M}_3$: $f(S|\mathbf{Z}=\mathbf{0},\mathbf{C})=f(S|\mathbf{Z}=\mathbf{0},\mathbf{C};\alpha^\dagger)$, $OR(\mathbf{Z},S,\mathbf{C})=OR(\mathbf{Z},S,\mathbf{C};\rho^\dagger)$, and $b_1(\mathbf{C})=b_1(\mathbf{C};\theta^{\dagger}_1)$.
\item $\mathcal{M}_4$: $f(\mathbf{Z}|S=0,\mathbf{C})=f(\mathbf{Z}|S=0,\mathbf{C};\tau^\dagger)$, $f(S|\mathbf{Z}=\mathbf{0},\mathbf{C})=f(S|\mathbf{Z}=\mathbf{0},\mathbf{C};\alpha^\dagger)$ and 
$OR(\mathbf{Z},S,\mathbf{C})=OR(\mathbf{Z},S,\mathbf{C};\rho^\dagger)$.
\end{itemize}
At one extreme, $\mathcal{M}_1$ requires a correct conditional mean model for outcome after removing the treatment effect: $E\{Y-\beta(\mathbf{A,C})S|\mathbf{Z},S,\mathbf{C}\}$. At the other, $\mathcal{M}_4$ requires a correct model for the joint conditional density $f(\mathbf{Z},S|\mathbf{C})$. However, we allow for certain combinations of restrictions on these laws. If our model for $f(S|\mathbf{Z},\mathbf{C})$ is correct, then under $\mathcal{M}_3$, assumptions on the impact of $\mathbf{Z}$ on the conditional mean of $Y$ given $\mathbf{C}$ in the reference population can be relaxed.

Our goal is to construct an estimator that is \textit{quadruply robust}, that is, unbiased if one of these four restrictions on the observed data (in addition to model $\mathcal{M}_{NEM}$) holds. To achieve this, in \ref{sm:est_strat} in the appendix we describe an estimation strategy for $\psi^\dagger$, which relies on estimators of $\rho^\dagger$, $\nu^\dagger$ and $\theta_1^\dagger$ that are themselves doubly robust. This means that we can obtain an unbiased estimator of $t(\mathbf{Z,C})$ even when $b_1(\mathbf{C})$ is misspecified (and vice versa). Likewise, we can obtain an unbiased estimate of $f(\mathbf{Z}|S,C)$ even when $f(S|\mathbf{Z}=\mathbf{0},C)$ is poorly modelled. Indeed, quadruple (rather than triple) robustness hinges on the specific parametrisation of the joint density $f(\mathbf{Z},S|\mathbf{C})$ using odds ratios.

\begin{theorem}\label{triply}
Under the union model $\mathcal{M}_{NEM}\cap(\mathcal{M}_1\cup \mathcal{M}_2 \cup\mathcal{M}_3\cup \mathcal{M}_4)$ and assuming standard regularity conditions hold, $\hat{\psi}_{MR-NEM}$ is a consistent and asymptotically normal (CAN) estimator of $\psi^\dagger$.
\end{theorem}
A proof is given in \ref{sm:proof_triply} of the appendix. A nonparametric estimator of the standard error for $\hat{\psi}_{MR-NEM}$ can be obtained either using a sandwich estimator, following standard M-estimation theory. Alternatively, one can use the nonparametric bootstrap. Suitable regularity conditions can be found e.g. in Appendix B of \citet{robins1994estimation}. In \ref{sec_lse} of the appendix, we obtain the semiparametric efficiency bound and the optimal choice of $m(\mathbf{C})$.

Our estimation theory has focused the reference population setting, to avoid repetition and because our estimators simplify in the case of negative control outcomes; see \ref{sm:nc_out} of the appendix. In this case, the resulting estimators simplify further and are closely related to the $g$-estimators of \citet{robins1994correcting}. They are doubly rather than quadruply robust, as fewer nuisance parameters are required to be estimated. Specifically, under NEM either the conditional mean models involving $\mathbf{Z}$ or $Y-W$ need to be correctly specified, in addition to the semiparametric structural model. In \ref{sm:att} we develop a semiparametric estimation theory for the marginal average treatment effect in the treated.

\section{Data analysis}
The Life Span Study of atomic bomb survivors in Japan has been influential among in understanding the impact of exposure to high levels of radiation on long-term health outcomes. However, initial analyses overlooked the potential role of confounding; it has been hypothesised that certain socioeconomic factors were associated both with location at the time of the bombing and the longer-term risk of cancer \citep{richardson2012lessons}. The dataset contains limited information on  potential confounders. On the other hand, there was data collected on residents who were not present in these cities at the time of the bombings, which forms a natural reference population.

We used a sample of 8,463 survivors who were 45-49 years old at the time of the bombings and who were followed up to December 31, 2000. These individuals were residents of Hiroshima or Nagasaki; out of the this sample, 1,787 people were away from the cities at the time of the bombing. A measure of high vs. low exposure to prompt radiation was based on weighted DS02 colon dose estimates, expressed as the weighted dose in gray (Gy); those in the reference population were assumed to have dose estimates equal to 0 Gy. For those who were present in the cities, an individual was considered to have had a high level of exposure if their dose estimates were above the median. The outcome of interest was age at death, measured in years. 

We chose city of residence as a potential instrument, since it was strongly predictive of the exposure and it seemed plausible that the its association with age at death under low radiation exposure could be stable across the two populations. However, we may not expect it to satisfy the exclusion restriction nor unconfoundedness.
We adjusted for sex as a covariate. We compared six estimators, described further in \ref{sm:add_est} of the appendix:
\begin{itemize}
    \item $\hat{\psi}_{TSLS}$: a two-stage least squares estimator that is unbiased under $\mathcal{M}_{NEM}\cap \mathcal{M}_1$.
    \item $\hat{\psi}_{g-Z}$: a $g$-estimator that is unbiased under $\mathcal{M}_{NEM}\cap \mathcal{M}_2$.
        \item $\hat{\psi}_{g-S}$: a $g$-estimator that is unbiased under $\mathcal{M}_{NEM}\cap \mathcal{M}_3$.
    \item $\hat{\psi}_{IPW}$: an inverse probability weighted estimator that is unbiased under $\mathcal{M}_{NEM}\cap \mathcal{M}_4$. 
     \item $\hat{\psi}_{MR}$: the multiply-robust estimator described in Section \ref{sec_mr} with $m(C)=1$
    \item  $\hat{\psi}_{MR-eff}$: the multiply-robust estimator with $m(C)$ set to the efficient choice.
\end{itemize}
Since both variables were binary in this case, it was possible to carry out a fully nonparametric analysis. Nevertheless, a treatment effect under NEM was chosen which excluded an interaction between $A$ and $C$. Further, to evaluate sensitivity to model misspecification, we considered five different model specifications for the nuisance parameters:
\begin{enumerate}
    \item All working  models for the nuisance parameters were saturated. 
    \item The models $OR(Z,S,C;\rho^\dagger)$ and  $t(Z,C;\nu^{\dagger})$ excluded interactions between $Z$ and $C$.
    \item The models $OR(Z,S,C;\rho^\dagger)$, $b_1(C;\theta^{\dagger}_1)$ excluded interactions between $Z$ and $C$.
    \item The models  $t(Z,C;\nu^{\dagger})$, $b_1(C;\theta^{\dagger}_1)$ excluded an interaction between $Z$ and $C$.
    \item All working models for the nuisance parameters excluded interactions between $Z$ and $C$.
\end{enumerate}
Varying the fitted models in this way enabled us to check the robustness of the estimators $\hat{\psi}_{TSLS}$, $\hat{\psi}_{g-Z}$, $\hat{\psi}_{g-S}$, $\hat{\psi}_{IPW}$, $\hat{\psi}_{MR}$ and $\hat{\psi}_{MR-eff}$ to departures from a less restrictive model. We note that if all true underlying models for the nuisance parameters include interactions, then in theory one can only construct estimators that are unbiased under the first and fourth specifications. To check sensitivity, we calculated the change in each estimator from estimates based on the first specification (saturated models), then scaling by the latter estimates.

\begin{table}
\caption{\label{tab_DA1}
Results from the Life Span study for the conditional effect in the treated. 95\% CI: 95\% confidence interval; \% Change: absolute percentage change from estimate based on saturated nuisance models.}
\centering
\begin{tabular}{lllll}
  \hline
Specification & Estimator & Estimate & $95\%$ CI & $\%$ Change \\ 
  \hline
All models saturated & $\hat{\psi}_{TSLS}$ & -1.81 & (-5.32,1.69) & 0 \\ 
   & $\hat{\psi}_{g-Z}$ & -1.81 & (-5.32,1.69) & 0 \\ 
   & $\hat{\psi}_{g-S}$ & -1.87 & (-5.38,1.63) & 0 \\ 
   & $\hat{\psi}_{IPW}$ & -1.82 & (-5.32,1.67) & 0 \\ 
   & $\hat{\psi}_{MR}$ & -1.82 & (-5.32,1.67) & 0 \\ 
   & $\hat{\psi}_{MR-eff}$ & -1.87 & (-5.37,1.64) & 0 \\ 
  $OR(Z,S,C;\rho^\dagger)$, $t(Z,C;\nu^{\dagger})$ restricted & $\hat{\psi}_{TSLS}$ & -1.82 & (-5.32,1.67) & 0.01 \\ 
   & $\hat{\psi}_{g-Z}$ & -2.35 & (-5.89,1.2) & 0.3 \\ 
   & $\hat{\psi}_{g-S}$ & -2.37 & (-5.92,1.18) & 0.26 \\ 
   & $\hat{\psi}_{IPW}$ & -4.04 & (-8.44,0.36) & 1.21 \\ 
   & $\hat{\psi}_{MR}$ & -1.82 & (-5.31,1.66) & $<$0.01 \\ 
   & $\hat{\psi}_{MR-eff}$ & -1.85 & (-5.34,1.65) & 0.01 \\ 
  $OR(Z,S,C;\rho^\dagger)$, $b_1(C;\theta^{\dagger}_1)$ restricted & $\hat{\psi}_{TSLS}$ & -1.82 & (-5.33,1.68) & 0.01 \\ 
   & $\hat{\psi}_{g-Z}$ & -6.29 & (-14.3,1.71) & 2.47 \\ 
   & $\hat{\psi}_{g-S}$ & -2.37 & (-5.92,1.18) & 0.27 \\ 
   & $\hat{\psi}_{IPW}$ & -4.04 & (-8.44,0.36) & 1.21 \\ 
   & $\hat{\psi}_{MR}$ & -1.84 & (-5.34,1.67) & 0.01 \\ 
   & $\hat{\psi}_{MR-eff}$ & -1.86 & (-5.35,1.64) & $<$0.01 \\ 
  $t(Z,C;\nu^{\dagger})$, $b_1(C;\theta^{\dagger}_1)$ restricted & $\hat{\psi}_{TSLS}$ & -1.83 & (-5.33,1.66) & 0.01 \\ 
   & $\hat{\psi}_{g-Z}$ & -1.82 & (-5.32,1.67) & 0.01 \\ 
   & $\hat{\psi}_{g-S}$ & -1.81 & (-5.31,1.69) & 0.03 \\ 
   & $\hat{\psi}_{IPW}$ & -1.82 & (-5.32,1.67) & 0 \\ 
   & $\hat{\psi}_{MR}$ & -1.82 & (-5.32,1.67) & 0 \\ 
   & $\hat{\psi}_{MR-eff}$ & -1.87 & (-5.37,1.64) & 0 \\ 
  All models restricted & $\hat{\psi}_{TSLS}$ & -1.83 & (-5.33,1.66) & 0.01 \\ 
   & $\hat{\psi}_{g-Z}$ & -2.35 & (-5.89,1.2) & 0.3 \\ 
   & $\hat{\psi}_{g-S}$ & -2.37 & (-5.92,1.18) & 0.27 \\ 
   & $\hat{\psi}_{IPW}$ & -4.04 & (-8.44,0.36) & 1.21 \\ 
   & $\hat{\psi}_{MR}$ & -1.84 & (-5.33,1.65) & 0.01 \\ 
   & $\hat{\psi}_{MR-eff}$ & -1.86 & (-5.36,1.63) & $<$0.01 \\ 
   \hline
\end{tabular}
\end{table}

The results of the analysis can be found in Table \ref{tab_DA1}. As might be expected, when all working models are saturated (except the exposure effect model), the estimates are in reasonable agreement. Fitting a linear model adjusted for radiation exposure, city of residence, sex and an interaction between the covariates yielded an exposure effext estimate of -0.32 (95\% CI: -0.82, 0.21). A two-stage least squares assumption under the assumption that city was a valid instrument yielded an estimate of 3.48 (95\% CI: 2.00, 4.96). Hence the estimates from our analysis tended to be arguably more plausible than from the alternative analyses, albeit with much wider confidence intervals. As one varies the model specifications, one can see the multiply robust estimators produce stable estimates; in contrast,  $\hat{\psi}_{g-Z}$, $\hat{\psi}_{g-S}$, $\hat{\psi}_{IPW}$ all appear to be highly sensitive to departures from fitting saturated models.

\section{Discussion}

In this article we have proposed strategies for identifying causal effects in studies prone to unmeasured confounding by leveraging both invalid instrumental variables and a reference population. Our notion of a reference population includes negative control populations (for whom the treatment has a null effect) as a special case. In this identification strategy, one must restrict either the conditional average treatment effect on the treated or the selection bias not to depend on $\mathbf{Z}$. Which restriction one chooses will depend to some extent on \textit{a priori} information on the choice of $\mathbf{Z}$, but will also be a matter of taste. Restrictions on the treatment effect are perhaps less appealing than in the standard instrumental variable set-up, where the instrument-outcome relationship is already constrained via the exclusion restriction. On the other hand, under NEM, the semiparametric model $\mathcal{M}_{NEM}$ is at least guaranteed to be correctly specified under the null hypothesis. A limitation of the semiparametric theory presented here is that it requires a correctly specified model for the conditional average treatment effect on the treated.  We have adopted this framework as it gives us greater flexibility in allowing for arbitrary $\mathbf{A}$ and $\mathbf{Z}$.

There are also interesting connections with the \textit{proximal inference} framework \citep{miao2018identifying,tchetgen2024introduction}. Here, a negative control outcome and exposure can together be leveraged to obtain nonparametric identification without a parallel trends assumption. In our case $\mathbf{Z}$ is a \textit{potentially invalid} negative control exposure. Moreover, when the reference population is a negative control population, the negative control outcome is simply given by $(1-S)Y$, but because it is not measured in the population of interest, proximal inference cannot be directly applied. Hence our proposal can be viewed as an alternative to proximal inference where the negative control outcome is missing in the population of interest and the negative control exposure is invalid.

Finally, one may wish to supplement our proposals with a sensitivity analysis. Attention may be given in particular to the crucial assumption of partial population exchangeability; one could then proceed by parametrising deviations from this condition. With multiple invalid instrumental variables, the average treatment effect on the treated may become over identified, and specification tests (like the Sargen-Hansen test) can then be employed to assess the validity of the identifying assumptions described here. If one has access to multiple reference populations, then one could compare associations between $\mathbf{Z}$ and $Y$ across the different populations. If associations of similar magnitudes are observed, this would provide support for the partial population exchangeability. This same idea motivates pre-trend tests in difference-in-difference designs. How to optimally combine or synthesise evidence from multiple invalid instrumental variables by leveraging negative controls is also an important area for future work.

\section*{Acknowledgement}
The authors would like to acknowledge support from the National Institutes of Health (grant  579679) and the Research Foundation Flanders (1222522N).

\vspace*{-10pt}

\appendix

\section{Additional information}\label{sm:add_inf}

\subsection{Linear structural equation models}\label{sm:lsem}

Consider the following models for the data-generating process:
\begin{align}\label{twfe}
E(Y^0|Z,S=s,\mathbf{C},U)&=\zeta_{0,s}+\zeta_{1,s} Z+\zeta_{2,s} \mathbf{C}+\zeta_{3,s} U \nonumber\\
E(U|Z,S=s,\mathbf{C})&=\upsilon_{0,s}+\upsilon_{1,s}Z+\upsilon_{2,s}\mathbf{C}
\end{align}
for $s=0,1$. For simplicity, we let $U$ be a single random variable. If we first consider a conditional parallel trends or `additive equi-confounding' condition \eqref{condition_pt}, note that \eqref{condition_pt} can be shown to hold if $\zeta_{3,0}=\zeta_{3,1}$ and 
\begin{align}\label{u_res}
&E(U|A=1,Z,S=1,\mathbf{C})-E(U|A=0,Z,S=1,\mathbf{C})\nonumber\\
&=E(U|A=1,Z,S=0,\mathbf{C})-E(U|A=0,Z,S=0,\mathbf{C})
\end{align}
with probability 1. Here, we are also assuming $Y^{0}\indep A |Z,S,\mathbf{C},U$ for $z=0,1$ (latent conditional exchangeability of treatment). Hence \eqref{condition_pt} can be justified under model \eqref{twfe} if the (conditional mean) association between outcome under control and the unmeasured confounder transports between the target and reference populations. The (conditional mean) association between $U$ and $A$ should also be equal across populations. 

Consider now a full population exchangeability condition \citep{dahabreh2019generalizing}:  $\forall z \in \{0,1\}$ and $\forall \mathbf{c}\in \mathcal{C}_0\cap \mathcal{C}_1$,
    \begin{align}\label{full_pop_ex}
E(Y^0|Z=z,S=1,\mathbf{C}=\mathbf{c})=E(Y^0|Z=z,S=0,\mathbf{C}=\mathbf{c}),
    \end{align}
then since 
\begin{align*}
E(Y^0|Z,S=s,\mathbf{C})&=\zeta_{0,s}+\zeta_{3,s}\upsilon_{0,s}+(\zeta_{1,s} +\zeta_{3,s}\upsilon_{1,s})Z+(\zeta_{2,s} +\zeta_{3,s}\upsilon_{2,s})\mathbf{C}
\end{align*}
for $s=0,1$, a sufficient condition for \eqref{full_pop_ex} to hold is that $\zeta_{i,0}=\zeta_{i,1}$ for $i=0,1,2,3$ and $\upsilon_{j,0}=\upsilon_{j,1}$ for $j=0,1,2$. Hence in the presence of an unmeasured variable $U$, we require all associations encoded by the coefficients in \eqref{twfe} to transport across populations. We note that beyond model \eqref{twfe}, Assumption \ref{partial_pe} is weaker than full population exchangeability since it only restricts the additive interaction between $Z$ and $S$ with respect to the conditional mean of $Y^0$. 

Moving now to the partial population exchangeability Assumption \ref{partial_pe} in the main manuscript, note that in contrast to the above, this requires only that $\zeta_{1,0}=\zeta_{1,1}$, $\zeta_{3,0}=\zeta_{3,1}$ and $\upsilon_{2,0}=\upsilon_{2,1}$. Some remarks below:

\begin{itemize}
\item Suppose that $Y^{0,z}\indep Z |S,\mathbf{C},U$ for $z=0,1$ (latent conditional exchangeability of the instrument). Then in \eqref{twfe}, $\zeta_{1,0}$ and $\zeta_{1,1}$ represent the direct effects of $Z$ on $Y^0$ in each population.
    \item Suppose that the previous condition is strengthened to $Y^{0,z}\indep Z |S,\mathbf{C}$ for $z=0,1$ with probability one (conditional exchangeability of the instrument). Then Assumption \ref{partial_pe} requires only that $\zeta_{1,0}=\zeta_{1,1}$. Hence in this case, the direct causal effect of $Z$ should transport, and no restrictions are placed on the unmeasured confounder. 
    \item Suppose that $\forall z$, $Y^{0,z}=Y^{0}$ (exclusion restriction). Then $\zeta_{1,0}=\zeta_{1,1}=0$ under the exclusion restriction. Hence Assumption \ref{partial_pe} restricts the (conditional mean) association between $Y^0$ and $U$, as was done to justify \eqref{condition_pt}, as well as the (conditional mean) association between $Z$ and $U$. 
    \item In this context, both additive equi-confounding \eqref{condition_pt} and  Assumption \ref{partial_pe} require that $\zeta_{3,0}=\zeta_{3,1}$. However, \eqref{condition_pt} is justified under \eqref{u_res} whereas Assumption \ref{partial_pe} requires that $\upsilon_{1,0}=\upsilon_{1,1}$ as well as $\zeta_{1,0}=\zeta_{1,1}$.
\end{itemize}

\subsection{Results for hypothesis testing}\label{sm:result_ht}

\begin{theorem}\label{theorem:test}
Under Assumptions \ref{ref}-\ref{partial_pe}, then  if 
\begin{align}\label{null_hyo}
E(Y^1-Y^0|A=1,Z=z,S=1,\mathbf{C}=\mathbf{c})=0
\end{align}
$\forall{z} \in \{0,1\}, \mathbf{c}\in \mathcal{C}_0\cap \mathcal{C}_1$, then it follows that 
\begin{align}
E\{Y-t(Z,\mathbf{c})|Z=1,S=1,\mathbf{C}=\mathbf{c}\}=E\{Y-t(Z,\mathbf{c})|Z=0,S=1,\mathbf{C}=\mathbf{c}\}.
\end{align}
\end{theorem}

A proof is given in section \ref{proof_test}.

\subsection{Negative control outcomes - identification}\label{sm:nco_id}

If we assume 
\begin{align}\label{psb_difference-in-differences_alt}
E(Y^0|A=1,Z,\mathbf{C})-E(Y^0|A=0, Z,\mathbf{C})=\gamma(\mathbf{C})
\end{align}
then we have
\begin{align*}
&E(Y^1-Y^0|A=1,Z=z,\mathbf{C}=\mathbf{c})\\
&=E(Y|A=1,z,\mathbf{c})-E(Y|A=0,z,\mathbf{c})\\
&\quad +\left\{\frac{E(Y|A=0,Z=1,\mathbf{c})-E(Y|A=0,Z=0,\mathbf{c})}{E(A|Z=1,\mathbf{c})-E(A|Z=0,\mathbf{c})}\right\}-\left\{\frac{E(W|Z=1,\mathbf{c})-E(W|Z=0,\mathbf{c})}
{E(A|Z=1,\mathbf{c})-E(A|Z=0,\mathbf{c})}\right\}.
\end{align*}

\subsection{Assumptions for a general instrumental variable model}\label{sm:add_assum}

\begin{assumption}{Reference population:}\label{ref_gen}
\[E(Y|\mathbf{Z},S=0,\mathbf{C})=E(Y^\mathbf{0}|\mathbf{Z},S=0,\mathbf{C})\]
with probability 1.
\end{assumption}

\begin{assumption}{Consistency:}\label{consist_gen}
$Y=Y^{\mathbf{A}}$ with probability 1 in individuals with $S=1$. 
\end{assumption}

\begin{assumption}{Relevance:}\label{relevance_gen}
$\mathbf{A} \nindep \mathbf{Z} | \mathbf{C},S=1$ with probability 1.
\end{assumption}

\begin{assumption}{Partial population exchangeability:}\label{partial_pe_gen}
\begin{align*}
&E(Y^\mathbf{0}|\mathbf{Z=z},S=1,\mathbf{C})-E(Y^\mathbf{0}|\mathbf{Z=0},S=1,\mathbf{C})\\
&=E(Y^\mathbf{0}|\mathbf{Z=z},S=0,\mathbf{C})-E(Y^\mathbf{0}|\mathbf{Z=0},S=0,\mathbf{C})
\end{align*}
with probability 1, for all $\mathbf{z}$ in the relevant support of $\mathbf{Z}$.
\end{assumption}

\subsection{Estimation under NSM}\label{sm:nsm}

We will proceed under a generalised version of the NSM assumption:
\begin{assumption}{No selection modification (NSM):}\label{psb_gen}\\
The generalised selection bias function
\[q(\mathbf{a},\mathbf{Z,C})\equiv E(Y^{\mathbf{0}}|\mathbf{A=a,Z},S=1,\mathbf{C})-E(Y^{\mathbf{0}}|\mathbf{A=0,Z},S=1,\mathbf{C})\]
satisfies the restriction $q(\mathbf{a,Z,C})=q(\mathbf{a,C})$ with probability 1.
\end{assumption}

Let us define
\begin{align*}
\epsilon&\equiv Y-\beta(\mathbf{A,Z,C})S-[q(\mathbf{A,Z,C})-E\{q(\mathbf{A,Z,C})|\mathbf{Z},S=1,\mathbf{C}\}]S\\
&\quad -t(\mathbf{Z,C})-b_{1}(\mathbf{C})S-b_{0}(\mathbf{C})
\end{align*}
where we redefine
\begin{align*}
b_1(\mathbf{C})\equiv& E\left(Y-\beta(\mathbf{A,Z,C})S-[q(\mathbf{A,Z,C})-E\{q(\mathbf{A,Z,C})|\mathbf{Z},S=1,\mathbf{C}\}]S|\mathbf{Z=0},S=1,\mathbf{C}\right)\\&-E(Y|\mathbf{Z=0},S=0,\mathbf{C})\\
b_0(\mathbf{C}) \equiv& E(Y|\mathbf{Z=0},S=0,\mathbf{C}).
\end{align*}

Then the orthocomplement to the nuisance tangent space under NSM can be obtained via the following extension of Theorem \ref{if_space}.
\begin{theorem}\label{if_space_b}
The orthocomplement to the nuisance tangent space under model $\mathcal{M}_{NSM}$ is given by
\begin{align*}
\bigg\{&\left\{S\kappa(\mathbf{A,Z,C})+\phi(\mathbf{Z},S,\mathbf{C}) \right\}\epsilon(\psi^\dagger)+S\phi(\mathbf{Z},S,\mathbf{C})[q(\mathbf{A,C})-E\{q(\mathbf{A,C})|\mathbf{Z},S=1,\mathbf{C}\}]\\&:\kappa(\mathbf{A,Z,C})\in \Gamma,\phi(\mathbf{Z},S,\mathbf{C})\in \Omega\bigg\}\cap L^0_2
\end{align*}
where 
\begin{align*}
\Gamma\equiv \left\{\kappa(\mathbf{A,Z,C}):E\{\kappa(\mathbf{A,Z,C})|\mathbf{Z},S=1,\mathbf{C}\}=E\{\kappa(\mathbf{A,Z,C})|\mathbf{A},S=1,\mathbf{C}\}=0\right\}.
\end{align*}
\end{theorem}
A proof is given in section \ref{sm:proof_if_nsm}.

Under the NSM assumption, we will utilise an additional parametric model $q(\mathbf{A,C};\omega^\dagger)$ for the selection bias function $q(\mathbf{A,C})$, where $q(\mathbf{A,C};\omega^\dagger)$ is a known function smooth in a finite dimensional parameter $\omega^\dagger$. A smooth model $f(\mathbf{A}|\mathbf{Z},S=1,\mathbf{C};\pi^\dagger)$ is postulated for $f(\mathbf{A}|\mathbf{Z},S=1,\mathbf{C})$; the event that model $f(\mathbf{A}|\mathbf{Z},S=1,\mathbf{C};\pi^\dagger)$ is correctly specified is denoted by $\mathcal{M}_a$. We will also consider alternative versions of the previous restrictions $\mathcal{M}_{1q}$ and $\mathcal{M}_{3q}$ on the data:
\begin{itemize}
\item $\mathcal{M}_{1q}$: $t(\mathbf{Z,C})=t(\mathbf{Z,C};\nu^{\dagger})$, $b_0(\mathbf{C})=b_0(\mathbf{C};\theta^{\dagger}_0)$,  $b_1(\mathbf{C})=b_1(\mathbf{C};\theta^{\dagger}_1)$ and $q(\mathbf{A,C})=q(\mathbf{A,C};\omega^\dagger)$.
\item $\mathcal{M}_{3q}$: $f(S|\mathbf{Z}=\mathbf{0},\mathbf{C})=f(S|\mathbf{Z}=\mathbf{0},\mathbf{C};\alpha^\dagger)$, $OR(\mathbf{Z},S,\mathbf{C})=OR(\mathbf{Z},S,\mathbf{C};\rho^\dagger)$, $b_1(\mathbf{C})=b_1(\mathbf{C};\theta^{\dagger}_1)$, and $q(\mathbf{A,C})=q(\mathbf{A,C};\omega^\dagger)$.
\end{itemize}
In order to estimate nuisance parameters and then obtain $\hat{\psi}_{MR-NSM}$ under NSM, one can use the strategy described in Section \ref{sm:est_strat}. 

\begin{theorem}\label{triply2}
Under the union $\mathcal{M}_{NSM}\cap \mathcal{M}_{a} \cap(\mathcal{M}_{1q}\cup \mathcal{M}_2 \cup\mathcal{M}_{3q}\cup\mathcal{M}_4)$ and standard regularity conditions, $\hat{\psi}_{MR-NSM}$ is a CAN estimator of $\psi^\dagger$.
\end{theorem}
A proof is given in section \ref{sm:proof_triply_nsm}. Hence, we can obtain a similar quadruple robustness result under the NSM assumption, although we now require a correctly specified propensity score model. Note that  a similar parametrisation of the joint density $f(\mathbf{A,Z}|S=1,\mathbf{C})$ based on the generalised odds ratio function could be used here. Nevertheless, since one needs to model $f(\mathbf{A}|\mathbf{Z},S=1,\mathbf{C})$ correctly in order to obtain an unbiased estimator, adopting such a parametrisation will not generally lead to greater robustness. Finally, since estimation under NSM requires consistent estimation of the propensity score, one approach would be to use flexible nonparametric/machine learning methods of $f(\mathbf{A}|\mathbf{Z},S=1,\mathbf{C})$ combined with cross-fitting, even if one were to then use parametric models for the other nuisance parameters. Under NSM, there appears to exist no closed-form expression for the efficient score for arbitrary $\mathbf{A}$ and $\mathbf{Z}$. 

\subsection{Algorithm for constructing cross-fit de-biased machine learning estimators}\label{sm:cross_fit}

If all nuisance parameter estimators converge at a rate faster than $n^{-1/4}$ and sample-splitting/cross-fitting is used, then the resulting estimators of $\psi^\dagger$ are RAL with variance equal to the variance of its influence function; see e.g. 
\citet{chernozhukov2018double} for a review.

In what follows, let $\mathbf{O}=(Y,\mathbf{A},\mathbf{Z},S,\mathbf{C})$ and  $\eta(\mathbf{O})$ refer to the nuisance parameters (at the truth). 
\begin{enumerate}
\item Split the sample into parts $I_k$ (that are each are of size $n_k=n/K$). Here, $K$ is an integer (we shall assume $n$ is a multiple of $K$). For each $I_k$, $I^c_{k}$ denotes the  indices that are not in $I_k$.
\item For each $k=1,...,k$, using $I^c_k$ only, estimate $\eta$ as $\hat{\eta}^c_k=\hat{\eta}((\mathbf{O}_i)_{i\in I^c_k})$. 
\item Construct $K$ estimators $\hat{\psi}_k, k=1,...,k$ of $\psi^\dagger$: under NEM, solve the equations 
\begin{align*}
0=\sum_{i\in I_k}\phi(\mathbf{Z}_i,S_i,\mathbf{C}_i;\hat{\eta}^c_{k})\epsilon_i^*(\psi^\dagger,\hat{\eta}^c_{k})
\end{align*}
for $\psi^\dagger$, for each $k=1,...,k$. Under NSM, solve the equations
\begin{align*}
&0=\begin{pmatrix}
&\sum_{i\in I_k}\phi(\mathbf{Z}_i,S_i,\mathbf{C}_i;\hat{\eta}^c_k)\epsilon^*_{i}(\psi^\dagger,\hat{\eta}^c_k)\\
&\sum_{i\in I_k}S_i\kappa(\mathbf{A}_i,\mathbf{Z}_i,\mathbf{C}_i;\hat{\eta}^c_k)\epsilon_{i}(\psi^\dagger,\hat{\eta}^c_k)
\end{pmatrix}
\end{align*}
for $\psi^\dagger$, for each $k=1,...,K$.
\item Take the average of the $K$ estimators of $\psi^\dagger$ to obtain $\hat{\psi}_{CF}$.
\item
Estimate the standard error for the cross-fit estimator of $\hat{\psi}_{CF}$ using a sandwich estimator, as described in \citet{chernozhukov2018double}. For example, under NEM, we estimate the asymptotic variance as $\hat{B}^{-1}\hat{V}(\hat{B}^{-1})^T$, where
\begin{align*}
\hat{V}&=\frac{1}{K}\sum^K_{k=1}\left[\frac{1}{n_k}\sum_{i\in I_k}\left\{\phi(\mathbf{Z}_i,S_i,\mathbf{C}_i;\hat{\eta}^c_{k})\epsilon^*_{i}(\hat{\psi}_{CF},\hat{\eta}^c_{k})\right\}\left\{\phi(\mathbf{Z}_i,S_i,\mathbf{C}_i;\hat{\eta}^c_{k})\epsilon^*_{i}(\hat{\psi}_{CF},\hat{\eta}^c_{k})\right\}^T\right]\\
\hat{B}&=-\frac{1}{K}\sum^K_{k=1}\left[\frac{1}{n_k}\sum_{i\in I_k} \phi(\mathbf{Z}_i,S_i,\mathbf{C}_i;\hat{\eta}^c_{k})\epsilon^{*'}_{i}(\hat{\psi}_{CF},\hat{\eta}^c_{k})\right]
\end{align*}
\end{enumerate}
and $\epsilon^{*'}(\psi,\eta)=\partial \epsilon^*(\psi,\eta)/\partial \psi$ is a row vector that is the length of the dimension of $\psi$.

As in other semiparametric regression problems, a complication is how to perform machine learning-based estimation of the nuisance parameters, given that they are not all conditional mean functions that can be estimated directly; see Section 4.2 of \citet{chernozhukov2018double} for further discussion. One option is to restrict to estimation methods that can more easily respect the structure of the semiparametric model  (generalised additive models, Lasso, deep neural networks etc); the other is to adopt the `localised de-biased machine learning' procedure described in \citet{kallus2019localized}, where e.g. an initial (biased) estimator of $\psi^\dagger$ could be used to construct estimators of other nuisance parameters. 

\subsection{Estimation strategy for $\psi^\dagger$ based on parametric working models}\label{sm:est_strat}

\begin{enumerate}
    \item Estimate $\tau^\dagger$ and $\alpha^\dagger$ as $\hat{\tau}$ and $\hat{\alpha}$ using maximum likelihood or M-estimation.
    \item Estimate $\rho^\dagger$ as $\hat{\rho}$ by solving the equations
\begin{align*}
 0=\sum^n_{i=1}&\left[e_1(\mathbf{Z}_i,\mathbf{C}_i)-\frac{E\{ e_1(\mathbf{Z}_i,\mathbf{C}_i)f(S_i=1|\mathbf{Z}_i,\mathbf{C}_i;\hat{\alpha},\rho^\dagger)f(S_i=0|\mathbf{Z}_i,\mathbf{C}_i;\hat{\alpha},\rho^\dagger)|\mathbf{C}_i;\hat{\tau},\hat{\alpha},\rho^\dagger\}}{E\{f(S_i=1|\mathbf{Z}_i,\mathbf{C}_i;\hat{\alpha},\rho^\dagger)f(S_i=0|\mathbf{Z}_i,\mathbf{C}_i;\hat{\alpha},\rho^\dagger)|\mathbf{C}_i;\hat{\tau},\hat{\alpha},\rho^\dagger\}}\right]\\
 & \times \{S_i-f(S_i=1|\mathbf{Z}_i,\mathbf{C}_i;\hat{\alpha},\rho^\dagger)\}
  \end{align*}
for $\rho^\dagger$ \citep{tchetgen2010doubly}, where $e_1(\mathbf{Z},\mathbf{C})$ is an arbitrary function of the same dimension of $\rho^\dagger$. 
    \item Estimate $\nu^\dagger$ and $\theta_0^\dagger$ as $\tilde{\nu}$ and $\hat{\theta}_0$ by solving the equations 
    \begin{align*}
 0=\sum^n_{i=1}& (1-S_i)
 \begin{pmatrix}
e_2(\mathbf{Z}_i,\mathbf{C}_i)\\
   e_3(\mathbf{C}_i)
 \end{pmatrix}\epsilon^*_i(\psi^\dagger,\nu^\dagger,\theta^\dagger_0)
  \end{align*}%
for $\nu^\dagger$ and $\theta_0^\dagger$, where $e_2(\mathbf{Z},\mathbf{C})$ and $e_3(\mathbf{C})$ are of the same dimension as $\nu^\dagger$ and $\theta^\dagger_0$ respectively. Re-estimate $\nu^\dagger$ as $\hat{\nu}$ by solving the equations
\begin{align*}
 0=\sum^n_{i=1}(1-S_i)[e_2(\mathbf{Z}_i,\mathbf{C}_i)-E\{e_2(\mathbf{Z}_i,\mathbf{C}_i)|S_i=0,\mathbf{C}_i;\hat{\tau}\}]\epsilon^*_i(\psi^\dagger,\nu^\dagger,\hat{\theta}_0,\theta_1^\dagger)
\end{align*}
for $\nu^\dagger$ \citep{robins1994correcting}.
\item Estimate $\theta_1^\dagger$ and $\psi^\dagger$ as $\hat{\theta}_1$ and $\tilde{\psi}$ respectively, by solving the equations
    \begin{align*}
 0=\sum^n_{i=1}& 
 \begin{pmatrix}
e_4(\mathbf{C}_i)\{S_i-E\{S_i|\mathbf{Z}_i,\mathbf{C}_i;\hat{\alpha},\hat{\rho}\}\\
   e_5(\mathbf{Z}_i,\mathbf{C}_i)\{S_i-E\{S_i|\mathbf{Z}_i,\mathbf{C}_i;\hat{\alpha},\hat{\rho}\}
 \end{pmatrix}\epsilon^*_i(\psi^\dagger,\hat{\nu},\hat{\theta}_0,\theta_1^\dagger)
  \end{align*}%
for $\theta_1^\dagger$ and $\psi^\dagger$, where $e_4(\mathbf{C})$ and $ e_5(\mathbf{Z},\mathbf{C})$ are of the same dimension as $\theta_1^\dagger$ and $\psi^\dagger$ respectively.
 \item Re-estimate $\psi^\dagger$ as $\hat{\psi}_{MR-NEM}$ by solving the equations
 \[0=\sum^n_{i=1}\phi(\mathbf{Z}_i,S_i,\mathbf{C}_i;\hat{\tau},\hat{\alpha},\hat{\rho})\epsilon_i^*(\psi ^\dagger,\hat{\nu},\hat{\theta})\]
 for $\psi^\dagger$.
\end{enumerate}

In the case that one is interested in inference under the NSM condition, there are additional nuisance parameters to consider. For example, $\pi^\dagger$ can be estimated as $\hat{\pi}$ via maximum likelihood; one can also obtain an estimator $\hat{\omega}$ by solving the equations 
    \begin{align*}
 0=\sum^n_{i=1}&
   e_6(A_i,\mathbf{C}_i)
  \epsilon_i(\psi^\dagger,\nu^\dagger,\theta^\dagger,\omega^\dagger,\hat{\pi})
  \end{align*}
for $\omega^\dagger$,  where $e_6(A,\mathbf{C})$ is of the same dimension as $\omega^\dagger$.  Then $\hat{\psi}_{MR-NSM}$ is obtained by solving the equations
\begin{align*}
 0=\sum^n_{i=1}&
 \begin{pmatrix}
\phi(\mathbf{Z}_i,S_i,\mathbf{C}_i;\hat{\tau},\hat{\alpha},\hat{\rho})\epsilon_i^*(\psi^\dagger,\hat{\nu},\hat{\theta})\\ 
S\kappa(\mathbf{A}_i,\mathbf{Z}_i,\mathbf{C}_i;\hat{\pi},\hat{\tau})\epsilon_{i}(\psi^\dagger,\hat{\nu},\hat{\theta},\hat{\omega},\hat{\pi})
 \end{pmatrix}
 \end{align*}
for $\psi^\dagger$.

\subsection{Local semiparametric efficiency}\label{sec_lse}

Once we have obtained the class of regular and asymptotically linear estimators of $\psi^\dagger$ under a semiparametric model, it remains to identify the optimal estimator within a given class. Under NEM, we are able to compute a closed form representation of the efficient score. Before doing this, we introduce some additional notation. 
Let us define \[\mu(\mathbf{Z},S,\mathbf{C})\equiv E\left\{\frac{\partial \beta(\mathbf{A,Z,C};\psi^\dagger)S}{\partial \psi}\bigg|_{\psi=\psi^\dagger}\bigg|\mathbf{Z},S,\mathbf{C}\right\},\]
\[P_{\sigma}(\mathbf{Z,C})\equiv \frac{f(S=1|\mathbf{Z,C})\sigma^{-2}(\mathbf{Z},1,\mathbf{C})f(S=0|\mathbf{Z,C})\sigma^{-2}(\mathbf{Z},0,\mathbf{C})}{f(S=1|\mathbf{Z,C})\sigma^{-2}(\mathbf{Z},1,\mathbf{C})+f(S=0|\mathbf{Z,C})\sigma^{-2}(\mathbf{Z},0,\mathbf{C})},\]
where $\sigma^{2}(\mathbf{Z},S,\mathbf{C})\equiv E\{\epsilon^{*2}(\psi^\dagger)|\mathbf{Z},S,\mathbf{C}\}$.
\begin{theorem}\label{pem_speff}
Under the semiparametric model $\mathcal{M}_{NEM}$ and with admissible independence density
\[
f^{\ddag }\left(\mathbf{Z},S|\mathbf{C}\right) =\left( \frac{1}{2}\right) ^{S}\left( \frac{1}{2}%
\right) ^{1-S}f(\mathbf{Z}|\mathbf{C})
\]
the optimal index function $r^{opt}_{0}(\mathbf{Z},S,\mathbf{C})$ in \begin{align}
\Omega=&\bigg\{\frac{f^\ddagger(\mathbf{Z},S|\mathbf{C})}{f(\mathbf{Z},S|\mathbf{C})}[r_0(\mathbf{Z},S,\mathbf{C})-E^\ddagger\{r_0(\mathbf{Z},S,\mathbf{C})|\mathbf{Z,C}\}\nonumber\\&-E^\ddagger\{r_0(\mathbf{Z},S,\mathbf{C})|S,\mathbf{C}\}+E^\ddagger\{r_0(\mathbf{Z},S,\mathbf{C})|\mathbf{C}\}]:\textrm{$r_0(\mathbf{Z},S,\mathbf{C})$ unrestricted}\bigg\}\label{aid_rep_zs}
\end{align}
is 
\[
r_{0}^{opt}(\mathbf{Z},S,\mathbf{C}) =\left\{\mu(\mathbf{Z},1,\mathbf{C})-\frac{E\left\{\mu(\mathbf{Z},1,\mathbf{C})P_{\sigma}(\mathbf{Z,C})|\mathbf{C}\right\}}{E\left\{P_{\sigma}(\mathbf{Z,C})|\mathbf{C}\right\}}\right\}2(-1)^{1-S}P_{\sigma }(\mathbf{Z,C}).
\]%
Furthermore, the efficient score is an unbiased estimating function in the union model $\mathcal{M}_{NEM}\cap(\mathcal{M}_1\cup \mathcal{M}_2 \cup\mathcal{M}_3 \cup \mathcal{M}_4)$ that is locally efficient at the intersection submodel $\mathcal{M}_{NEM}\cap \mathcal{M}_1\cap \mathcal{M}_2 \cap\mathcal{M}_3 \cap \mathcal{M}_4$.
\end{theorem}
A proof is given in \ref{sm:proof_speff}. We emphasise that the choice of admissible independence densities is not an assumption on the data generating mechanism. Rather, it allows us to simplify the complex general expression for the efficient score into an equation of the form \eqref{aid_rep_zs}. Similar to the standard instrumental variable set-up, efficient estimation generally requires modelling the conditional mean of $\mathbf{A}$ given $\mathbf{Z}$ and $\mathbf{C}$ in those with $S=1$. However, even if this model is misspecified, the resulting estimating of $\psi^\dagger$ will remain consistent (although it will no longer generally be efficient). In order to further simplify the above result, we revisit Example \ref{bin_bin}.

\begin{corollary}\label{pem_speff_bin}
If $\mathbf{Z}$ is binary, then the optimal choice $m_{opt}(\mathbf{C})$ of $m(\mathbf{C})$ in the estimating function (\ref{binary_ee}) is 
\begin{align*}
&m_{{opt}}(\mathbf{C})=E\left[
\left\{\frac{(-1)^{Z+S}}{f(Z,S|\mathbf{C})}\right\}^2\sigma^{2}(Z,S,\mathbf{C})
\bigg| \mathbf{C}\right]^{-1}\{\mu(Z=1,S=1,\mathbf{C})-\mu(Z=0,S=1,\mathbf{C})\}.
\end{align*}
If $\sigma^{2}(Z,S,\mathbf{C})=\sigma^{2}$, then $m_{{opt}}(\mathbf{C})$ reduces to $w_0(\mathbf{C})^{-1}\sigma^{-2}\{\mu(Z=1,S=1,\mathbf{C})-\mu(Z=0,S=1,\mathbf{C})\}$, where 
\[w_0(\mathbf{C})=\frac{1}{f(Z=1,S=1|\mathbf{C})}+\frac{1}{f(Z=1,S=0|\mathbf{C})}+\frac{1}{f(Z=0,S=1|\mathbf{C})}+\frac{1}{f(Z=0,S=0|\mathbf{C})}.\]
\end{corollary}
Note that the weight function $w_0(\mathbf{C})^{-1}$ will automatically give high weight given in the estimating equations to strata of the data where there is most overlap in terms of the distribution $f(Z,S|\mathbf{C})$. In studies where the inverse weights $1/f(Z,S|\mathbf{C})$ are extreme, we might expect an efficient estimator to thus perform better than other multiply robust estimators both in terms of precision as well as stability and finite sample bias.

\subsection{Negative control outcomes - estimation}\label{sm:nc_out}

Let us define (or redefine)
\begin{align*}
q(\mathbf{a,Z,C})\equiv&E(Y^\mathbf{0}|\mathbf{A=a},\mathbf{Z},\mathbf{C})-E(Y^\mathbf{0}|\mathbf{A=0},\mathbf{Z},\mathbf{C})\\
t(\mathbf{z,C})\equiv&E(Y^\mathbf{0}|\mathbf{Z=z},\mathbf{C})-E(Y^\mathbf{0}|\mathbf{Z=0},\mathbf{C})\\
b^*_1(\mathbf{C})\equiv&E(Y^\mathbf{0}|\mathbf{Z=0},\mathbf{C})\\
b_0(\mathbf{C})\equiv&E(W^\mathbf{0}|\mathbf{Z=0},\mathbf{C}).
\end{align*}
Under slightly altered versions of Assumptions \ref{ref_gen} and \ref{partial_pe_gen}, $t(\mathbf{z,C})$ is identified as 
\[t(\mathbf{z,C})=E(w|\mathbf{Z=z},\mathbf{C})-E(w|\mathbf{Z=0},\mathbf{C}).\]
The following result is a corollary of Theorem \ref{if_space}.
\begin{corollary}\label{cor_did}
 The (conditional) density of $S$ is degenerate and  $f(\mathbf{Z}|S=1,\mathbf{C})=f(\mathbf{Z}|S=0,\mathbf{C})=f(\mathbf{Z}|\mathbf{C})$, such that the nuisance tangent space is smaller.
Under the semiparametric model implied by restriction (\ref{mod_res}) and the NEM Assumption \ref{pem_gen}, the estimating function for $\psi^\dagger$ implied by Theorem \ref{if_space} simplifies to
\begin{align}\label{pem_did_ef}
[\phi(\mathbf{Z},\mathbf{C})-E\{\phi(\mathbf{Z},\mathbf{C})|\mathbf{C}\}][Y-W-\beta(\mathbf{A,C};\psi^\dagger)-\{b^*_1(\mathbf{C})-b_0(\mathbf{C})\}]
\end{align}
where $\phi(\mathbf{Z,C})$ is arbitrary. Also, if $q(\mathbf{a,z,c})=q(\mathbf{a,c})$, then the estimating functions for $\psi^\dagger$ derived under the NSM Assumption \ref{psb_gen} reduce to 
\begin{align}\label{psb_did_ef}
\begin{pmatrix}
[\phi(\mathbf{Z},\mathbf{C})-E\{\phi(\mathbf{Z},\mathbf{C})|\mathbf{C}\}][Y-W-\beta(\mathbf{A,Z,C};\psi^\dagger)-\{b^*_1(\mathbf{C})-b_0(\mathbf{C})\}]\\
\kappa(\mathbf{A,Z,C})\left(Y-\beta(\mathbf{A,Z,C};\psi^\dagger)-[q(\mathbf{A,C})-E\{q(\mathbf{A,C})|\mathbf{Z,C}\}]-t(\mathbf{Z,C})-b^*_1(\mathbf{C})\right)
\end{pmatrix}
\end{align}
where $\phi(\mathbf{Z,C})$ is arbitrary and $\kappa(\mathbf{A,Z,C})$ satisfies $E\{\kappa(\mathbf{A,Z,C})|\mathbf{A,C}\}=E\{\kappa(\mathbf{A,Z,C})|\mathbf{Z,C}\}=0$.
\end{corollary}

With negative control outcomes, following results in Section \ref{sec_mr} of the main paper, note that estimators of $\psi^\dagger$ based on the estimating functions given in Corollary \ref{cor_did} are doubly robust. Specifically, an estimator based on equation (\ref{pem_did_ef}) is unbiased if either i) $f(\mathbf{Z}|\mathbf{C})$ or ii) $b^*_1(\mathbf{C})$ and $b_0(\mathbf{C})$ (or their difference) are consistently estimated. Similarly, an estimator based on equation (\ref{psb_did_ef}) is unbiased if, in addition to $f(\mathbf{A}|\mathbf{Z,C})$, either i) $f(\mathbf{Z}|\mathbf{C})$ or ii) $q(\mathbf{A,C})$,  $t(\mathbf{Z,C})$ $b^*_1(\mathbf{C})$ and $b_0(\mathbf{C})$ are consistently estimated. Hence, with a valid negative control outcome, our estimators are expected to be additionally robust and efficient relative to estimators with a valid reference population.

\subsection{The average treatment effect in the treated}\label{sm:att}

 It may be that interest is in the marginal rather than the conditional ATT. Suppose that $\mathbf{A}$ and $\mathbf{Z}$ are binary; then in the case that $\beta(Z)=\beta$, previous methods can be used due to the collapsibility of the linear link function,  but one may not wish to invoke this restriction.  It follows from the identification results in Section 2 that under the NEM Assumption \ref{pem},  \[\psi^*\equiv E(Y^1-Y^0|A=1,S=1)=E\{E(Y^1-Y^0|A=1,Z, S=1,\mathbf{C})|A=1,S=1\}\] is identified as 
\begin{align*}
\psi^*=\int \beta(\mathbf{c}) dF(\mathbf{c}|A=1, S=1)
\end{align*}
where 
\begin{align*}
\beta(\mathbf{C})&\equiv \frac{\delta^Y(\mathbf{C})-t(1, \mathbf{C})}{\delta^A(\mathbf{C})}\\
\delta^Y(\mathbf{C})&\equiv E(Y|Z=1,S=1,\mathbf{C})-E(Y|Z=0,S=1,\mathbf{C})\\
\delta^A(\mathbf{C})&\equiv f(A=1|Z=1,S=1,\mathbf{C})-f(A=1|Z=0,S=1, \mathbf{C}).
\end{align*}
We therefore have a closed form representation of $\psi^*$ as a functional of the observed data distribution,  Since $\psi^*$ is a pathwise differentiable parameter,  we can then obtain nonparametric inference on $\psi$,  by obtaining the efficient influence function under a nonparametric model.  The following results may be of independent interest with respect to the literature of nonparametric inference in the IV model \citep{wang2018bounded}.

\begin{theorem}\label{np_inf}
 In the nonparametric model defined by Assumptions \ref{consist}-\ref{pem},  the efficient influence function is equal to:
\begin{align*}
EIF_1(\psi^*)=&\frac{f(A=1,S=1|\mathbf{C})(-1)^{Z+S}}{f(A=1,S=1)\delta^A(\mathbf{C})f(Z,S|\mathbf{C})}\\&\times[Y-\beta(\mathbf{C})\{A-f(A=1|Z=0,S=1,\mathbf{C})\}S-t(1,\mathbf{C})Z-E(Y|Z=0,S,\mathbf{C})]\\
&+\frac{1}{f(A=1,S=1)}AS\left\{\beta(\mathbf{C})-\psi^*\right\}.
\end{align*}
\end{theorem}
A proof can be found in section \ref{sm:proof_np}. One can use these results to create estimators of $\psi^*$,  after obtaining estimates of the nuisance parameters involved.

\section{Proofs}\label{sm:proofs}

\subsection{Proof of Proposition \ref{main_prop}}\label{sm:proof_prop}
\begin{proof}
For any $z\in \{0,1\}$ and $\mathbf{c}\in \mathcal{C}_0\cap \mathcal{C}_1$, 
\begin{align*}
E(Y|Z=z,S=0,\mathbf{C}=\mathbf{c})&=E(Y^{0}|Z=z,S=0,\mathbf{C}=\mathbf{c}) & \textrm{(Assumption \ref{ref})} \\
&=E(Y^{0,z}|Z=z,S=0,\mathbf{C}=\mathbf{c}) &\textrm{(Instrument consistency)}\\
&=E(Y^{0,z}|S=0,\mathbf{C}=\mathbf{c}) & (\textrm{Unconfoundedness})\\
&=E(Y^{0}|S=0,\mathbf{C}=\mathbf{c}) & (\textrm{Exclusion restriction})
\end{align*}
where $E(Y^{0}|S=0,\mathbf{C}=\mathbf{c})$ is not a function of $z$.
\end{proof}

\subsection{Proof of Theorem \ref{iden}}\label{sm:proof_iden}

\begin{proof}
Note that $\forall z\in \{0,1\}$ and $\forall \mathbf{c} \in \mathcal{C}_0\cap \mathcal{C}_1$
\begin{align*}
E(Y|A=a,z,S=1,\mathbf{c})=&\beta(a,z,\mathbf{c})+\gamma(z,\mathbf{c})\{a-f(A=1|Z=z,S=1,\mathbf{c})\}\\&+t(z,\mathbf{c})+E(Y^{0}|Z=0,S=1,\mathbf{c})
\end{align*}
\citep{tchetgen2013alternative}. 
To obtain identification, the right hand side must be reduced by two parameters. 

First, 
\begin{align*}
t(1,\mathbf{c})&=
E(Y^0|Z=1,S=0,\mathbf{C}=\mathbf{c})-E(Y^0|Z=0,S=0,\mathbf{C}=\mathbf{c}) &\textrm{(Assumption \ref{partial_pe})}\\
&=E(Y|Z=1,S=0,\mathbf{C}=\mathbf{c})-E(Y|Z=0,S=0,\mathbf{C}=\mathbf{c}). &\textrm{(Assumption \ref{ref})} 
\end{align*}
Then we have 
\begin{align*}
&E(Y|A=a,z,S=1,\mathbf{c})\\&=\beta(a,z,\mathbf{c})+\gamma(z,\mathbf{c})\{a-f(A=1|Z=z,S=1,\mathbf{c})\}\\&\quad +E(Y|Z=z,S=0,\mathbf{C}=\mathbf{c})-E(Y|Z=0,S=0,\mathbf{C}=\mathbf{c})+E(Y^{0}|Z=0,S=1,\mathbf{c}).
\end{align*}
Then either NEM or NSM can be applied to remove a further degree-of-freedom. The remaining results follow via rearranging the above formula as an expression for $\beta(1,z,\mathbf{c})$; see also the Appendix of \citet{richardson2021bespoke} where the second result (under NEM) is shown. 
\end{proof}

\subsection{Proof of Theorem \ref{if_space}}\label{sm:proof_if}

\begin{proof}
We introduce some additional notation:
\begin{align*}
\epsilon_{1}&\equiv Y-\beta(\mathbf{A,Z,C})-[q(\mathbf{A,Z,C})-E\{q(\mathbf{A,Z,C})|\mathbf{Z},S=1,\mathbf{C}\}]-b^*_{1}(\mathbf{C})-t(\mathbf{Z,C})\\
\epsilon^*_{1}&\equiv \epsilon_{1}+[q(\mathbf{A,Z,C})-E\{q(\mathbf{A,Z,C})|\mathbf{Z},S=1,\mathbf{C}\}]\\
\epsilon_{0}&\equiv Y-b_{0}(\mathbf{C})-t(\mathbf{Z,C})
\end{align*}
where $b^*_1(\mathbf{C})=b_1(\mathbf{C})+b_0(\mathbf{C})$. Note that in order to obtain identification,  we enforce that $\beta(\mathbf{A,Z,C})=\beta(\mathbf{A,C})$.

The likelihood for a given observation can be written as 
\[f(\mathbf{O})=f(\epsilon_0|\mathbf{Z},S=0,\mathbf{C})^{(1-S)}\left\{f(\epsilon_1|\mathbf{A},\mathbf{Z},S=1,\mathbf{C})f(\mathbf{A}|\mathbf{Z},S=1,\mathbf{C})\right\}^{S}f(\mathbf{Z},S|\mathbf{C})f(\mathbf{C}).\]
Then we will consider the parametric submodel:
\[f_r(\mathbf{O})=f_r(\epsilon_{0_r}|\mathbf{Z},S=0,\mathbf{C})^{(1-S)}\left\{f_r(\epsilon_{1_r}|\mathbf{A,Z},S=1,\mathbf{C})f_r(A|\mathbf{Z},S=1,\mathbf{C})\right\}^{S}f_r(\mathbf{Z},S|\mathbf{C})f_r(S,\mathbf{C})\]
which varies in the direction of $q_r(A,\mathbf{Z})$,  $b_{0_r}(\mathbf{C})$,  $b^*_{1_r}(\mathbf{C})$, $t_r(\mathbf{Z})$,  $f_r(\mathbf{A}|\mathbf{Z},S=1,\mathbf{C})$,  $f_r(\mathbf{Z},S|\mathbf{C})$ and $f_r(\mathbf{C})$.  Here,
\begin{align*}
\epsilon_{1_r}&=Y-\beta(\mathbf{A,Z,C})-[q_r(\mathbf{A,Z,C})-E_r\{q_r(\mathbf{A,Z,C})|\mathbf{Z},S=1,\mathbf{C}\}]-b^*_{1_r}(\mathbf{C})-t_r(\mathbf{Z,C})\\
\epsilon_{0_r}&=Y-b_{0_r}(\mathbf{C})-t_r(\mathbf{Z,C}).
\end{align*}

The nuisance tangent space $\lambda_{nuis}$ under $\mathcal{M}_{NEM}$ can be characterised as 
\[\lambda_{nuis}=\lambda_{nuis_1}\oplus\lambda_{nuis_2}\oplus\lambda_{nuis_3}\oplus\lambda_{nuis_4}\oplus\lambda_{nuis_5}\oplus\lambda_{nuis_6}\oplus\lambda_{nuis_7}\oplus\lambda_{nuis_8}\oplus\lambda_{nuis_9}\]
where
\begin{align*}
\Lambda_{nuis_1}\equiv \bigg\{&Sd_1(\epsilon_1,\mathbf{A,Z,C}):\\& E\{ d_1(\epsilon_1,\mathbf{A,Z,C})|\mathbf{A,Z},S=1,\mathbf{C}\}=E\{\epsilon_1d_1(\epsilon_1,\mathbf{A,Z,C})|\mathbf{A,Z},S=1,\mathbf{C}\}=0 \bigg\}\cap L_2^0\\
\Lambda_{nuis_2}\equiv \bigg\{&(1-S)d_2(\epsilon_0,\mathbf{Z,C}): E\{d_2(\epsilon_0,\mathbf{Z,C})|\mathbf{Z},S=0,\mathbf{C}\}=E\{ \epsilon_0d_2(\epsilon_1,\mathbf{Z})|\mathbf{Z},S=0,\mathbf{C}\}=0 \bigg\}\cap L_2^0\\
\Lambda_{nuis_3}\equiv \bigg\{&d_3(\mathbf{Z},S,\mathbf{C}):E\{d_3(\mathbf{Z},S,\mathbf{C})|\mathbf{C}\}=0 \bigg\}\cap L_2^0\\
\Lambda_{nuis_4}\equiv \bigg\{&d_4(\mathbf{C}):E\{d_4(\mathbf{C})\}=0 \bigg\}\cap L_2^0\\
\Lambda_{nuis_5}\equiv \bigg\{&S[d_5(\mathbf{A,Z,C})-E\{d_5(\mathbf{A,Z,C})|\mathbf{Z},S=1,\mathbf{C}\}]f'_{\epsilon_{1}}: \textrm{$d_5(\mathbf{A,Z,C})$ unrestricted}\bigg\}\cap L_2^0\\
\Lambda_{nuis_6}\equiv \bigg\{&Sd_6(A, \mathbf{Z,C})+Sf'_{\epsilon_{1}}\int d_6(A^*, \mathbf{Z,C})q(A^*,\mathbf{Z,C})dF(A^*|\mathbf{Z},S=1,\mathbf{C}):\\&  E\{d_6(A, \mathbf{Z,C})|\mathbf{Z},S=1,\mathbf{C}\}=0\bigg\}\cap L_2^0\\
\Lambda_{nuis_7}\equiv \bigg\{&Sd_7(\mathbf{C})f'_{\epsilon_{1}}:  \textrm{$d_7(\mathbf{C})$ unrestricted}\bigg\}\cap L_2^0\\
\Lambda_{nuis_8}\equiv \bigg\{&(1-S)d_8(\mathbf{C})f'_{\epsilon_0}: \textrm{$d_8(\mathbf{C})$ unrestricted}\bigg\}\cap L_2^0\\
\Lambda_{nuis_9}\equiv \bigg\{&Sd_9(\mathbf{Z,C})f'_{\epsilon_{1}}+(1-S)d_9(\mathbf{Z,C})f'_{\epsilon_0}: \textrm{$d_9(\mathbf{C})$ unrestricted}\bigg\}\cap L_2^0
\end{align*}
where $f'_{\epsilon_{1}}$ is the derivative of $f(\epsilon_{1}|A,\mathbf{Z},S=1,\mathbf{C})$ w.r.t. $\epsilon_1$ and likewise $f'_{\epsilon_{0}}$ is the derivative of $f(\epsilon_0|\mathbf{Z},S=0,\mathbf{C})$ w.r.t. $\epsilon_0$.

Using standard results on the restricted mean model e.g. \citet{tsiatis2007semiparametric}, we have that 
\begin{align*}
\Lambda^\perp_{nuis_1}\cap \Lambda^\perp_{nuis_3}\cap  \Lambda^\perp_{nuis_4} = \bigg\{&S\epsilon_1c_1(\mathbf{A,Z,C})+Sc_2(\mathbf{A,Z},\mathbf{C}):\\&\textrm{$c_1(\mathbf{A,Z,C})$ unrestricted},E\{c_2(\mathbf{A,Z},\mathbf{C})|\mathbf{Z},S=1,\mathbf{C}\}=0 \bigg\}\cap L_2^0
\end{align*}
Then to find elements in this space that are orthogonal to $\Lambda_{nuis_5}$,  we must find the elements in $\Lambda^\perp_{nuis_1}\cap \Lambda^\perp_{nuis_3}\cap  \Lambda^\perp_{nuis_4}$ that satisfy:
\begin{align*}
0&=E\left(S\{\epsilon_1c_1(\mathbf{A,Z,C})+c_2(\mathbf{A,Z},\mathbf{C})\}[d_5(\mathbf{A,Z,C})-E\{d_5(\mathbf{A,Z,C})|\mathbf{Z},S=1,\mathbf{C}\}] f'_{\epsilon_{1}}\right)\\
&=E\left(Sc_1(\mathbf{A,Z,C})[d_5(\mathbf{A,Z,C})-E\{d_5(\mathbf{A,Z,C})|\mathbf{Z},S=1,\mathbf{C}\}]\right).
\end{align*}
It follows that elements of $Sc_1(\mathbf{A,Z,C})$ that will satisfy the equality are of the form $S\phi_1(\mathbf{Z,C})$, where $\phi_1(\mathbf{Z,C})$ are unrestricted.  One can also show that the elements 
\[S\epsilon_1[\phi_1(\mathbf{Z,C})-E\{\phi_1(\mathbf{Z,C})|S=1,\mathbf{C}\}]+Sc_2(\mathbf{A,Z},\mathbf{C})\]
are orthogonal to $\Lambda_{nuis_7}$.

Next, in considering $\Lambda_{nuis_6}$,  we need elements that satisfy:
\begin{align*}
0=E\bigg\{&S\left(\epsilon_1[\phi_1(\mathbf{Z,C})-E\{\phi_1(\mathbf{Z,C})|S=1,\mathbf{C}\}]+c_2(\mathbf{A,Z},\mathbf{C})\right)\\
&\times [d_6(A, \mathbf{Z,C})-E\{Sd_6(A, \mathbf{Z,C})q(\mathbf{A,Z,C})|\mathbf{Z},S=1,\mathbf{C}\}]f'_{\epsilon_{1}}\bigg\}
\end{align*}
Then note that
\begin{align*}
&E\bigg\{S\epsilon_1[\phi_1(\mathbf{Z,C})-E\{\phi_1(\mathbf{Z,C})|S=1,\mathbf{C}\}]d_6(\mathbf{A,Z,C})\bigg\}
=0\\
&E\bigg\{-Sc_2(\mathbf{A,Z},\mathbf{C})E\{d_6(A, \mathbf{Z,C})q(\mathbf{A,Z,C})|\mathbf{Z},S=1,\mathbf{C}\}f'_{\epsilon_{1}} \bigg\}
=0
\end{align*}
We are left with the restriction that 
\begin{align*}
0=E\bigg(& Sc_2(\mathbf{A,Z},\mathbf{C})d_6(\mathbf{A,Z,C})\\&-S\epsilon_1[\phi_1(\mathbf{Z,C})-E\{\phi_1(\mathbf{Z,C})|S=1,\mathbf{C}\}]E\{d_6(A, \mathbf{Z,C})q(\mathbf{A,Z,C})|\mathbf{Z},S=1,\mathbf{C}\}f'_{\epsilon_{1}} \bigg)\\
=E\bigg(& Sd_6(\mathbf{A,Z,C})\\&\times \left(c_2(\mathbf{A,Z},\mathbf{C})-[\phi_1(\mathbf{Z,C})-E\{\phi_1(\mathbf{Z,C})|S=1,\mathbf{C}\}][q(\mathbf{A,Z,C})-E\{q(\mathbf{A,Z,C})|\mathbf{Z},S=1,\mathbf{C}\}]\right)\bigg)
\end{align*}
and so 
\[c_2(\mathbf{A,Z},\mathbf{C})=[\phi_1(\mathbf{Z,C})-E\{\phi_1(\mathbf{Z,C})|S=1,\mathbf{C}\}][q(\mathbf{A,Z,C})-E\{q(\mathbf{A,Z,C})|\mathbf{Z},S=1,\mathbf{C}\}].\]
Hence,  so far we have shown that 
\begin{align*}
&\Lambda^\perp_{nuis_1}\cap \Lambda^\perp_{nuis_3}\cap \Lambda^\perp_{nuis_4} \cap \Lambda^\perp_{nuis_5} \cap \Lambda^\perp_{nuis_6} \cap \Lambda^\perp_{nuis_7}\\ &=\bigg\{S\epsilon^*_1[\phi_1(\mathbf{Z,C})-E\{\phi_1(\mathbf{Z,C})|S=1,\mathbf{C}\}]:
\textrm{$\phi_1(\mathbf{Z,C})$ unrestricted} \bigg\}\cap L_2^0.
\end{align*}

Considering elements in $\Lambda^\perp_{nuis_2}\cap\Lambda^\perp_{nuis_8}$,  then using previous reasoning,  we have that 
\[\Lambda^\perp_{nuis_2}\cap\Lambda^\perp_{nuis_8}=\left\{(1-S)\epsilon_0[\phi_2(\mathbf{Z,C})-E\{\phi_2(\mathbf{Z,C})|\mathbf{C},S=0\}]; \textrm{$\phi_2(\mathbf{Z,C})$ unrestricted}\right\}.\]
It is also straightforward to show that elements in the space $\Lambda^\perp_{nuis_2}\cap\Lambda^\perp_{nuis_8}$ are orthogonal to those in the space 
$\Lambda^\perp_{nuis_1}\cap \Lambda^\perp_{nuis_3}\cap \Lambda^\perp_{nuis_4} \cap \Lambda^\perp_{nuis_5} \cap \Lambda^\perp_{nuis_6} \cap \Lambda^\perp_{nuis_7}$,  and therefore 
\begin{align*}
&\Lambda^\perp_{nuis_1}\cap\Lambda^\perp_{nuis_2}\cap \Lambda^\perp_{nuis_3}\cap \Lambda^\perp_{nuis_4} \cap \Lambda^\perp_{nuis_5} \cap \Lambda^\perp_{nuis_6} \cap \Lambda^\perp_{nuis_7} \cap \Lambda^\perp_{nuis_8}\\
&=\bigg\{S\epsilon^*_1[\phi_1(\mathbf{Z,C})-E\{\phi_1(\mathbf{Z,C})|S=1,\mathbf{C}\}]\\&\quad+(1-S)\epsilon_0[\phi_2(\mathbf{Z,C})-E\{\phi_2(\mathbf{Z,C})|\mathbf{C},S=0\}]: \textrm{$\phi_1(\mathbf{Z,C}),\phi_2(\mathbf{Z,C})$ unrestricted}\bigg\}.
\end{align*}
It remains to find elements in this space that are orthogonal to $\Lambda^\perp_{nuis_9}$; these must satisfy
\begin{align*}
0=&E\bigg\{\bigg(S\epsilon^*_1[\phi_1(\mathbf{Z,C})-E\{\phi_1(\mathbf{Z,C})|S=1,\mathbf{C}\}]+(1-S)\epsilon_0[\phi_2(\mathbf{Z,C})-E\{\phi_2(\mathbf{Z,C})|S=0,\mathbf{C}\}]
\bigg)\\&\times\{Sd_9(\mathbf{Z,C})f'_{\epsilon_{1}}+(1-S)d_9(\mathbf{Z,C})f'_{\epsilon_{0}}\}\bigg)\\
=&E\bigg\{\bigg(S[\phi_1(\mathbf{Z,C})-E\{\phi_1(\mathbf{Z,C})|S=1,\mathbf{C}\}] +(1-S)[\phi_2(\mathbf{Z,C})-E\{\phi_2(\mathbf{Z,C})|S=0,\mathbf{C}\}]\bigg)
d_9(\mathbf{Z,C})\bigg\}.
\end{align*}
It follows from \citet{tchetgen2010doubly} that the space of elements 
\[S[\phi_1(\mathbf{Z,C})-E\{\phi_1(\mathbf{Z,C})|S=1,\mathbf{C}\}]+(1-S)[\phi_2(\mathbf{Z,C})-E\{\phi_2(\mathbf{Z,C})|\mathbf{C},S=0\}]\]
that satisfy this equality can be represented as 
\[\Omega=\{\phi(\mathbf{Z},S,\mathbf{C}):E\{\phi(\mathbf{Z},S,\mathbf{C})|\mathbf{Z,C}\}=E\{\phi(\mathbf{Z},S,\mathbf{C})|S,\mathbf{C}\}=0\}.\] 
Then the main result follows. 
\end{proof}

\subsection{Proof of Theorem \ref{triply}}\label{sm:proof_triply}
\begin{proof}

We will first show that
\[E\{\phi(\mathbf{Z},S,\mathbf{C};\tau^\dagger,\alpha^\dagger,\rho^\dagger)\epsilon^*(\psi^\dagger,\nu^\dagger,\theta^\dagger)\}\]
evaluated at the limiting (rather than estimated) values is an unbiased estimating function under the union model. Under model $\mathcal{M}_{NEM}\cap\mathcal{M}_1$, by the law of iterated expectation,
\begin{align*}
&E\{\phi(\mathbf{Z},S,\mathbf{C};\tau^\dagger,\alpha^\dagger,\rho^\dagger)\epsilon^*(\psi^\dagger,\nu^\dagger,\theta^\dagger)\}\\
&=E\bigg(\phi(\mathbf{Z},S,\mathbf{C};\tau^\dagger,\alpha^\dagger,\rho^\dagger)[b_0(\mathbf{C})+b_1(\mathbf{C})S+t(\mathbf{Z,C})-b_0(\mathbf{C};\theta_0^\dagger)-b_1(\mathbf{C};\theta_1^\dagger)S-t(\mathbf{Z,C};\nu^\dagger)\\&\quad+q(\mathbf{A,Z,C})-E\{q(\mathbf{A,Z,C})|\mathbf{Z},S=1,\mathbf{C}\}]\bigg)\\
&=E\left[\phi(\mathbf{Z},S,\mathbf{C};\tau^\dagger,\alpha^\dagger,\rho^\dagger)\{b_0(\mathbf{C})+b_1(\mathbf{C})S+t(\mathbf{Z,C})-b_0(\mathbf{C};\theta_0^\dagger)-b_1(\mathbf{C};\theta_1^\dagger)S-t(\mathbf{Z,C};\nu^\dagger)\}\right]=0.
\end{align*}

Under model $\mathcal{M}_{NEM}\cap\mathcal{M}_2$, looking back to the final line of the previous expression,
\begin{align*}
&E\left[\phi(\mathbf{Z},S,\mathbf{C};\tau^\dagger,\alpha^\dagger,\rho^\dagger)\{b_0(\mathbf{C})+b_1(\mathbf{C})S+t(\mathbf{Z,C})-b_0(\mathbf{C};\theta_0^\dagger)-b_1(\mathbf{C};\theta_1^\dagger)S-t(\mathbf{Z,C};\nu^\dagger)\}\right]\\
&=E\left[\phi(\mathbf{Z},S,\mathbf{C};\tau^\dagger,\alpha^\dagger,\rho^\dagger)\{b_0(\mathbf{C})+b_1(\mathbf{C})S-b_0(\mathbf{C};\theta_0^\dagger)-b_1(\mathbf{C};\theta_1^\dagger)S\}\right]\\
&=E\left[E\{\phi(\mathbf{Z},S,\mathbf{C};\tau^\dagger,\alpha^\dagger,\rho^\dagger)|S,\mathbf{C}\}\{b_0(\mathbf{C})+b_1(\mathbf{C})S-b_0(\mathbf{C};\theta_0^\dagger)-b_1(\mathbf{C};\theta_1^\dagger)S\}\right]=0
\end{align*}
where the final equality follows since $f(\mathbf{Z}|S,\mathbf{C})=f(\mathbf{Z}|S,\mathbf{C};\tau^\dagger,\rho^\dagger)$.

Under model $\mathcal{M}_{NEM}\cap\mathcal{M}_3$,
\begin{align*}
&E\left[\phi(\mathbf{Z},S,\mathbf{C};\tau^\dagger,\alpha^\dagger,\rho^\dagger)\{b_0(\mathbf{C})+b_1(\mathbf{C})S+t(\mathbf{Z,C})-b_0(\mathbf{C};\theta_0^\dagger)-b_1(\mathbf{C};\theta_1^\dagger)S-t(\mathbf{Z,C};\nu^\dagger)\}\right]\\&=E\left[E\{\phi(\mathbf{Z},S,\mathbf{C};\tau^\dagger,\alpha^\dagger,\rho^\dagger)|\mathbf{Z,C}\}\{b_0(\mathbf{C})+t(\mathbf{Z,C})-b_0(\mathbf{C};\theta_0^\dagger)-t(\mathbf{Z,C};\nu^\dagger)\}\right]=0
\end{align*}
where the final equality follows since $f(S|\mathbf{Z,C})=f(S|\mathbf{Z,C};\alpha^\dagger,\rho^\dagger)$.

Finally, under model $\mathcal{M}_{NEM}\cap\mathcal{M}_4$,
\begin{align*}
&E\left[\phi(\mathbf{Z},S,\mathbf{C};\tau^\dagger,\alpha^\dagger,\rho^\dagger)\{b_0(\mathbf{C})+b_1(\mathbf{C})S+t(\mathbf{Z,C})-b_0(\mathbf{C};\theta_0^\dagger)-b_1(\mathbf{C};\theta_1^\dagger)S-t(\mathbf{Z,C};\nu^\dagger)\}\right]\\&=
E\left[E\{\phi(\mathbf{Z},S,\mathbf{C};\tau^\dagger,\alpha^\dagger,\rho^\dagger)|S,\mathbf{C}\}\{b_1(\mathbf{C})S-b_1(\mathbf{C};\theta_1^\dagger)S\}\right]\\
&\quad+E\left[E\{\phi(\mathbf{Z},S,\mathbf{C};\tau^\dagger,\alpha^\dagger,\rho^\dagger)|\mathbf{Z,C}\}\{t(\mathbf{Z,C})-t(\mathbf{Z,C};\nu^\dagger)\}\right]\\
&\quad+E\left\{\phi(\mathbf{Z},S,\mathbf{C};\tau^\dagger,\alpha^\dagger,\rho^\dagger)\{b_0(\mathbf{C})-b_0(\mathbf{C};\theta_0^\dagger)\}\right\}=0.
\end{align*}

Note furthermore that
\begin{itemize}
    \item $f(\mathbf{Z}|S,\mathbf{C};\hat{\tau},\hat{\rho})$ is a CAN estimator of $f(\mathbf{Z}|S,\mathbf{C})$ in model $\mathcal{M}_2\cup\mathcal{M}_4$.
    \item $f(S|\mathbf{Z,C};\hat{\alpha},\hat{\rho})$ is a CAN estimator of $f(S|\mathbf{Z,C})$ in model $\mathcal{M}_3\cup\mathcal{M}_4$.
    \item $t(\mathbf{Z,C};\hat{\nu})$ is a CAN estimator of $t(\mathbf{Z,C})$ in model $\mathcal{M}_{NEM}\cap(\mathcal{M}_1\cup\mathcal{M}_2)$.
    \item $b_1(\mathbf{C};\hat{\theta}_1)$ is a CAN estimator of $b_1(\mathbf{C})$ in model $\mathcal{M}_{NEM}\cap(\mathcal{M}_1\cup\mathcal{M}_3)$.
    \item  $b_0(\mathbf{C};\hat{\theta}_0)$ is a CAN estimator of $b_0(\mathbf{C})$ in model $\mathcal{M}_{NEM}\cap(\mathcal{M}_1\cup\mathcal{M}_2\cup\mathcal{M}_3)$.
\end{itemize}
These results follow from \citet{robins1992estimating} and \citet{tchetgen2010doubly}. This completes the proof.
\end{proof}

\subsection{Proof of Theorem \ref{theorem:test}}\label{proof_test}

\begin{proof}
The conditional mean $E(Y|A=a,Z=z,S=1,\mathbf{C}=\mathbf{c})$ is equal to 
\begin{align*}
&E(Y|A=a,Z=z,S=1,\mathbf{c})\\&=
E(Y^{1}-Y^{0}|A=1,Z=z,S=1,\mathbf{c})a\\
&\quad +\gamma(z,\mathbf{c})\{a-Pr(A=1|Z=z,S=1,\mathbf{c})\}\\
&\quad +t(z,\mathbf{c})+E(Y^{0}|Z=0,S=1,\mathbf{c}).
\end{align*}
Note that 
\begin{align*}
E(Y^1-Y^0|A=1,Z=z,S=1,\mathbf{C}=\mathbf{c})&=0 \quad &(\textrm{Under \eqref{null_hyo}})\\
E(Y|Z=z,S=0,\mathbf{c})-E(Y|Z=0,S=0,\mathbf{c})&=t(z,\mathbf{c}). \quad &(\textrm{Assumptions }\ref{ref} \textrm{ and } \ref{partial_pe})
\end{align*}
Then after recentering $Y$ and averaging over $A$,
\begin{align*}
&E\{Y-E(Y|Z=z,S=0,\mathbf{c})+E(Y|Z=0,S=0,\mathbf{c})|Z=z,S=1,\mathbf{c}\}\\
&=E(Y^{0}|Z=0,S=1,\mathbf{c})
\end{align*}
which is not a function of $z$.
\end{proof}

\subsection{Proof of Theorem \ref{if_space_b}}\label{sm:proof_if_nsm}

\begin{proof}
Many of the steps in in this proof follow along the lines of the previous proof and can therefore be omitted for brevity. We note however that NSM Assumption (\ref{psb_gen}) imposes different restrictions on the nuisance tangent space, such that we redefine $\Lambda_{nuis_5}$ as
\begin{align*}
\Lambda_{nuis_5}= \bigg\{&S[d_5(\mathbf{A,C})-E\{d_5(\mathbf{A,C})|\mathbf{Z},S=1,\mathbf{C}\}]f'_{\epsilon_{1}}:d_5(\mathbf{A,C}) \textrm{ unrestricted} \bigg\}\cap L_2^0.
\end{align*}

Along the lines of the proof of Theorem 4 in \citet{tchetgen2013alternative}, one can show that 
\begin{align*}
&\Lambda^\perp_{nuis_1}\cap \Lambda^\perp_{nuis_3}\cap \Lambda^\perp_{nuis_4} \cap \Lambda^\perp_{nuis_5} \cap \Lambda^\perp_{nuis_6} \cap \Lambda^\perp_{nuis_7}\\ &=\bigg\{S\epsilon_1[\kappa(\mathbf{A,Z,C})+\phi_1(\mathbf{Z,C})-E\{\phi_1(\mathbf{Z,C})|S=1,\mathbf{C}\}]\\
&\quad\quad+S[q(\mathbf{A,Z,C})-E\{q(\mathbf{A,Z,C})|\mathbf{Z},S=1,\mathbf{C}\}][\phi_1(\mathbf{Z,C})-E\{\phi_1(\mathbf{Z,C})|S=1,\mathbf{C}\}]:\\&\quad\quad 
\textrm{$\kappa(\mathbf{A,Z,C}) \in \Gamma$,  $\phi_1(\mathbf{Z,C})$ unrestricted} \bigg\}\cap L_2^0
\end{align*}
and as in the previous proof, 
\[\Lambda^\perp_{nuis_2}\cap\Lambda^\perp_{nuis_8}=\left\{(1-S)\epsilon_0[\phi_2(\mathbf{Z,C})-E\{\phi_2(\mathbf{Z,C})|\mathbf{C},S=0\}]: \textrm{$\phi_2(\mathbf{Z,C})$ unrestricted}\right\}\]
Then one can find the elements 
\begin{align*}
&S\epsilon_1[\kappa(\mathbf{A,Z,C})+\phi_1(\mathbf{Z,C})-E\{\phi_1(\mathbf{Z,C})|S=1,\mathbf{C}\}]\\
&+[q(\mathbf{A,Z,C})-E\{q(\mathbf{A,Z,C})|\mathbf{Z},S=1,\mathbf{C}\}][\phi_1(\mathbf{Z,C})-E\{\phi_1(\mathbf{Z,C})|S=1,\mathbf{C}\}]\\
&+(1-S)\epsilon_0[\phi_2(\mathbf{Z,C})-E\{\phi_2(\mathbf{Z,C})|\mathbf{C},S=0\}]
\end{align*}
that are orthogonal to $\lambda_{nuis_9}$ as in the previous proof, and the main result follows.

\end{proof}

\subsection{Proof of Theorem \ref{triply2}}\label{sm:proof_triply_nsm}

\begin{proof}
If we consider the estimating function
\[S\kappa(\mathbf{A,Z,C};\pi^\dagger,\tau^\dagger,\rho^\dagger)\epsilon(\psi^\dagger,\nu^\dagger,\theta^\dagger,\omega^\dagger,\pi^\dagger)+\phi(\mathbf{Z},S,\mathbf{C};\tau^\dagger,\alpha^\dagger,\rho^\dagger)\epsilon^*(\psi^\dagger,\nu^\dagger,\theta^\dagger)\]
then it follows from the previous proof that 
\[E\{\phi(\mathbf{Z},S,\mathbf{C};\tau^\dagger,\alpha^\dagger,\rho^\dagger)\epsilon^*(\psi^\dagger,\nu^\dagger,\theta^\dagger)\}=0\]
under the union model. 

Under model $\mathcal{M}_{NSM}\cap\mathcal{M}_a\cap\mathcal{M}_{1q}$,
\begin{align*}
&E\{S\kappa(\mathbf{A,Z,C};\pi^\dagger,\tau^\dagger,\rho^\dagger)\epsilon(\psi^\dagger,\nu^\dagger,\theta^\dagger,\omega^\dagger,\pi^\dagger)\}\\
&=E\bigg(S\kappa(\mathbf{A,Z,C};\pi^\dagger,\tau^\dagger,\rho^\dagger)[b_0(\mathbf{C})+b_1(\mathbf{C})S+t(\mathbf{Z,C})-b_0(\mathbf{C};\theta_0^\dagger)-b_1(\mathbf{C};\theta_1^\dagger)S-t(\mathbf{Z,C};\nu^\dagger)\\&\quad+q(\mathbf{A,C})-E\{q(\mathbf{A,C})|\mathbf{Z},S=1,\mathbf{C}\}-q(\mathbf{A,C};\omega^\dagger)+E\{q(\mathbf{A,C};\omega^\dagger)|\mathbf{Z},S=1,\mathbf{C};\pi^\dagger\}]\bigg)=0
\end{align*}
Under model $\mathcal{M}_{NSM}\cap\mathcal{M}_a\cap\mathcal{M}_{2}$,
\begin{align*}
&E\bigg(S\kappa(\mathbf{A,Z,C};\pi^\dagger,\tau^\dagger,\rho^\dagger)[b_0(\mathbf{C})+b_1(\mathbf{C})S+t(\mathbf{Z,C})-b_0(\mathbf{C};\theta_0^\dagger)-b_1(\mathbf{C};\theta_1^\dagger)S-t(\mathbf{Z,C};\nu^\dagger)\\&\quad+q(\mathbf{A,C})-E\{q(\mathbf{A,C})|\mathbf{Z},S=1,\mathbf{C}\}-q(\mathbf{A,C};\omega^\dagger)+E\{q(\mathbf{A,C};\omega^\dagger)|\mathbf{Z},S=1,\mathbf{C};\pi^\dagger\}]\bigg)\\
&=E\bigg(S\kappa(\mathbf{A,Z,C};\pi^\dagger,\tau^\dagger,\rho^\dagger)[b_0(\mathbf{C})+b_1(\mathbf{C})S-b_0(\mathbf{C};\theta_0^\dagger)-b_1(\mathbf{C};\theta_1^\dagger)S-\\&\quad+q(\mathbf{A,C})-E\{q(\mathbf{A,C})|\mathbf{Z},S=1,\mathbf{C}\}-q(\mathbf{A,C};\omega^\dagger)+E\{q(\mathbf{A,C};\omega^\dagger)|\mathbf{Z},S=1,\mathbf{C};\pi^\dagger\}]\bigg)\\
&=E\left[SE\{\kappa(\mathbf{A,Z,C};\pi^\dagger,\tau^\dagger,\rho^\dagger)|S=1,\mathbf{C}\}\{b_0(\mathbf{C})+b_1(\mathbf{C})S-b_0(\mathbf{C};\theta_0^\dagger)-b_1(\mathbf{C};\theta_1^\dagger)S\}\right]\\
&\quad+E\left[SE\{\kappa(\mathbf{A,Z,C};\pi^\dagger,\tau^\dagger,\rho^\dagger)|\mathbf{A},S=1,\mathbf{C}\}\{q(\mathbf{A,C})-q(\mathbf{A,C};\omega^\dagger)\}\right]\\&
\quad-E\bigg(SE\{\kappa(\mathbf{A,Z,C};\pi^\dagger,\tau^\dagger,\rho^\dagger)|\mathbf{Z},S=1,\mathbf{C}\}[E\{q(\mathbf{A,C})|\mathbf{Z},S=1,\mathbf{C}\}\\&\quad-E\{q(\mathbf{A,C};\omega^\dagger)|\mathbf{Z},S=1,\mathbf{C};\pi^\dagger\}]\bigg)\\
&=0
\end{align*}
where the final equality follows since $f(\mathbf{A,Z}|S=1,\mathbf{C})=f(\mathbf{A,Z}|S=1,\mathbf{C};\pi^\dagger,\tau^\dagger,\rho^\dagger)$.

Under model $\mathcal{M}_{NSM}\cap\mathcal{M}_a\cap\mathcal{M}_{3q}$,
\begin{align*}
&E\bigg(S\kappa(\mathbf{A,Z,C};\pi^\dagger,\tau^\dagger,\rho^\dagger)[b_0(\mathbf{C})+b_1(\mathbf{C})S+t(\mathbf{Z,C})-b_0(\mathbf{C};\theta_0^\dagger)-b_1(\mathbf{C};\theta_1^\dagger)S-t(\mathbf{Z,C};\nu^\dagger)\\&\quad+q(\mathbf{A,C})-E\{q(\mathbf{A,C})|\mathbf{Z},S=1,\mathbf{C}\}-q(\mathbf{A,C};\omega^\dagger)+E\{q(\mathbf{A,C};\omega^\dagger)|\mathbf{Z},S=1,\mathbf{C};\pi^\dagger\}]\bigg)\\
&=E\left[S\kappa(\mathbf{A,Z,C};\pi^\dagger,\tau^\dagger,\rho^\dagger)\{b_0(\mathbf{C})+t(\mathbf{Z,C})-b_0(\mathbf{C};\theta_0^\dagger)-b_1(\mathbf{C};\theta_1^\dagger)S-t(\mathbf{Z,C};\nu^\dagger)\}\right]\\
&=E\left[SE\{\kappa(\mathbf{A,Z,C};\pi^\dagger,\tau^\dagger,\rho^\dagger)|\mathbf{Z},S=1,\mathbf{C}\}\{b_0(\mathbf{C})+t(\mathbf{Z,C})-b_0(\mathbf{C};\theta_0^\dagger)-b_1(\mathbf{C};\theta_1^\dagger)S-t(\mathbf{Z,C};\nu^\dagger)\}\right]\\&=0
\end{align*}
where the final equality follows since $f(\mathbf{A}|\mathbf{Z},S=1,\mathbf{C})=f(\mathbf{A}|\mathbf{Z},S=1,\mathbf{C};\pi^\dagger)$.

Under model $\mathcal{M}_{NSM}\cap\mathcal{M}_a\cap\mathcal{M}_{4}$, it follows from previous arguments that
\begin{align*}
&E\bigg(S\kappa(\mathbf{A,Z,C};\pi^\dagger,\tau^\dagger,\rho^\dagger)[b_0(\mathbf{C})+b_1(\mathbf{C})S+t(\mathbf{Z,C})-b_0(\mathbf{C};\theta_0^\dagger)-b_1(\mathbf{C};\theta_1^\dagger)S-t(\mathbf{Z,C};\nu^\dagger)\\&\quad+q(\mathbf{A,C})-E\{q(\mathbf{A,C})|\mathbf{Z},S=1,\mathbf{C}\}-q(\mathbf{A,C};\omega^\dagger)+E\{q(\mathbf{A,C};\omega^\dagger)|\mathbf{Z},S=1,\mathbf{C};\pi^\dagger\}]\bigg)=0
\end{align*}
since $f(\mathbf{A,Z}|S=1,\mathbf{C})=f(\mathbf{A,Z}|S=1,\mathbf{C};\pi^\dagger,\tau^\dagger,\rho^\dagger)$.

\end{proof}

\subsection{Proof of Theorem \ref{pem_speff}}\label{sm:proof_speff}

\begin{proof}
Under a restricted moment model where $t(\mathbf{Z,C})$, $b_0(\mathbf{C})$, $b_1(\mathbf{C})$ and $q(\mathbf{A,Z,C})$ are known, it follows e.g. from \citet{robins1994correcting} that the efficient score for $\psi^\dagger$ is equal to 
\[U_{RM-eff}=\mu(\mathbf{Z},S,\mathbf{C})\sigma^{-2}(\mathbf{Z},S,\mathbf{C})\epsilon^*(\psi^\dagger).\]
It also follows from the proofs of Theorem 1 and 4 in \citet{vansteelandt2008multiply} that the nuisance tangent space under $\mathcal{M}_{NEM}$ obtained in Theorem \ref{if_space} can be equivalently represented as 
\[\bigg\{\{e_0(\mathbf{C})+e_1(S,\mathbf{C})+e_2(\mathbf{Z,C})\}\sigma^{-2}(\mathbf{Z},S,\mathbf{C})\epsilon^*(\psi^\dagger):e_0\in \mathcal{C}_0,e_1\in \mathcal{C}_1,e_2\in \mathcal{C}_2\bigg\}\cap L_2^0\]
where $\mathcal{C}_0=\{e_0(\mathbf{C})\}\cap\mathcal{H}$, $\mathcal{C}_1=\{e_1(S,\mathbf{C})\}\cap\mathcal{H}$,  $\mathcal{C}_2=\{e_2(\mathbf{Z,C})\}\cap\mathcal{H}$ and $\mathcal{H}$ is the Hilbert space of functions of $\mathbf{Z},S,\mathbf{C}$ with inner product given by $E\{\sigma^{-2}(\mathbf{Z},S,\mathbf{C})e_1(\mathbf{Z},S,\mathbf{C})e_2(\mathbf{Z},S,\mathbf{C})\}$ for all $e_1, e_2$.

Then the efficient score under model $\mathcal{M}_{NEM}$ is equal to the following projection 
\begin{align*}
&\Pi_{L_2^0} \left(U_{RM-eff}|[\{e_0(\mathbf{C})+e_1(S,\mathbf{C})+e_2(\mathbf{Z,C})\}\sigma^{-2}(\mathbf{Z},S,\mathbf{C})\epsilon^*(\psi^\dagger)]^{\perp}:e_0\in \mathcal{C}_0,e_1\in \mathcal{C}_1,e_2\in \mathcal{C}_2 \right)\\
&=\Pi_{\mathcal{H}} \left(\mu(\mathbf{Z},S,\mathbf{C})|\{e_0(\mathbf{C})+e_1(S,\mathbf{C})+e_2(\mathbf{Z,C})\}^{\perp}:e_0\in \mathcal{C}_0,e_1\in \mathcal{C}_1,e_2\in \mathcal{C}_2 \right)\sigma^{-2}(\mathbf{Z},S,\mathbf{C})\epsilon^*(\psi^\dagger)
\end{align*}
where $\Pi_{L_2^0} (\cdot|\cdot)$ is the orthogonal projection operator in $L_2^0$, and $\Pi_{H} (\cdot|\cdot)$ is similarly defined w.r.t. $\mathcal{H}$. \citet{vansteelandt2008multiply} obtain a closed form representation of the projection in $\mathcal{H}$ of any function $D \in \mathcal{H}$ onto $(\mathcal{C}_0\cup \mathcal{C}_1 \cup \mathcal{C}_2)^\perp$ as 
\begin{align*}
J\left(D\right)-\frac{E\{D\sigma^{-2}(\mathbf{Z},S,\mathbf{C})\tilde{S}|\mathbf{C}\}}{E\{ \sigma^{-2}(\mathbf{Z},S,\mathbf{C})\tilde{S}^2|\mathbf{C}\}}\tilde{S}-J\left(\frac{E \{D\sigma^{-2}(\mathbf{Z},S,\mathbf{C})|\mathbf{Z,C}\}}{E\{ \sigma^{-2}(\mathbf{Z},S,\mathbf{C})|\mathbf{Z,C}\}}\right)
\end{align*}
where we define the $J(\cdot)$ operator on function $D^*$ as
\[J(D^*)=D^*-\frac{E\{D^*\sigma^{-2}(\mathbf{Z},S,\mathbf{C})|\mathbf{C}\}}{E\{\sigma^{-2}(\mathbf{Z},S,\mathbf{C})|\mathbf{C}\}}\]
and
\begin{align*}
\tilde{S}&\equiv S-\frac{E\{S\sigma^{-2}(\mathbf{Z},S,\mathbf{C})|\mathbf{Z,C}\}}{E\{\sigma^{-2}(\mathbf{Z},S,\mathbf{C})|\mathbf{Z,C}\}}\\
&=S-\frac{f(S=1|\mathbf{Z,C})\sigma^{-2}(\mathbf{Z},1,\mathbf{C})}{f(S=1|\mathbf{Z,C})\sigma^{-2}(\mathbf{Z},1,\mathbf{C})+f(S=0|\mathbf{Z,C})\sigma^{-2}(\mathbf{Z},0,\mathbf{C})}.
\end{align*}
Applying this result, we obtain the efficient score as 
\begin{align*}
&\left[J(\mu(\mathbf{Z},S,\mathbf{C}))-\frac{E\{\mu(\mathbf{Z},S,\mathbf{C})\sigma^{-2}(\mathbf{Z},S,\mathbf{C})\tilde{S}|\mathbf{C}\}}{E\{ \sigma^{-2}(\mathbf{Z},S,\mathbf{C})\tilde{S}^2|\mathbf{C}\}}\tilde{S}-J\left(\frac{E \{\mu(\mathbf{Z},S,\mathbf{C})\sigma^{-2}(\mathbf{Z},S,\mathbf{C})|\mathbf{Z,C}\}}{E\{ \sigma^{-2}(\mathbf{Z},S,\mathbf{C})|\mathbf{Z,C}\}}\right)\right]\\
&\times \sigma^{-2}(\mathbf{Z},S,\mathbf{C})\epsilon^*(\psi^\dagger)\\
&=\left[\mu(\mathbf{Z},S,\mathbf{C})-\frac{E \{\mu(\mathbf{Z},S,\mathbf{C})\sigma^{-2}(\mathbf{Z},S,\mathbf{C})|\mathbf{Z,C}\}}{E\{ \sigma^{-2}(\mathbf{Z},S,\mathbf{C})|\mathbf{Z,C}\}}-\frac{E\{\mu(\mathbf{Z},S,\mathbf{C})\sigma^{-2}(\mathbf{Z},S,\mathbf{C})\tilde{S}|\mathbf{C}\}}{E\{ \sigma^{-2}(\mathbf{Z},S,\mathbf{C})\tilde{S}^2|\mathbf{C}\}}\tilde{S}\right]\\
&\times \sigma^{-2}(\mathbf{Z},S,\mathbf{C})\epsilon^*(\psi^\dagger)
\end{align*}
where we use that 
\begin{align*}
&E\{\sigma^{-2}(\mathbf{Z},S,\mathbf{C})|\mathbf{C}\}^{-1}E\left[\sigma^{-2}(\mathbf{Z},S,\mathbf{C})\frac{E \{\mu(\mathbf{Z},S,\mathbf{C})\sigma^{-2}(\mathbf{Z},S,\mathbf{C})|\mathbf{Z,C}\}}{E\{ \sigma^{-2}(\mathbf{Z},S,\mathbf{C})|\mathbf{Z,C}\}}\bigg|\mathbf{C}\right]\\
&=E\{\sigma^{-2}(\mathbf{Z},S,\mathbf{C})|\mathbf{C}\}^{-1}E\left[E \{\mu(\mathbf{Z},S,\mathbf{C})\sigma^{-2}(\mathbf{Z},S,\mathbf{C})|\mathbf{Z,C}\}|\mathbf{C}\right]\\
&=E\{\sigma^{-2}(\mathbf{Z},S,\mathbf{C})|\mathbf{C}\}^{-1}E \{\mu(\mathbf{Z},S,\mathbf{C})\sigma^{-2}(\mathbf{Z},S,\mathbf{C})|\mathbf{C}\}.
\end{align*}

To further simply the previous expression for the efficient score,  
note that
\begin{align*}
\mu(\mathbf{Z},S,\mathbf{C})&=\mu(\mathbf{Z},1,\mathbf{C})S,\\
\frac{E\{\mu(\mathbf{Z},S,\mathbf{C})\sigma^{-2}(\mathbf{Z},S,\mathbf{C})|\mathbf{Z,C}\}}{E\{\sigma^{-2}(\mathbf{Z},S,\mathbf{C})|\mathbf{Z,C}\}}
&=\frac{\mu(\mathbf{Z},1,\mathbf{C})f(S=1|\mathbf{Z,C})\sigma^{-2}(\mathbf{Z},1,\mathbf{C})}{f(S=1|\mathbf{Z,C})\sigma^{-2}(\mathbf{Z},1,\mathbf{C})+f(S=0|\mathbf{Z,C})\sigma^{-2}(\mathbf{Z},0,\mathbf{C})}
\end{align*}
and thus
\[\mu(\mathbf{Z},S,\mathbf{C})-\frac{E\{\mu(\mathbf{Z},S,\mathbf{C})\sigma^{-2}(\mathbf{Z},S,\mathbf{C})|\mathbf{Z,C}\}}{E\{\sigma^{-2}(\mathbf{Z},S,\mathbf{C})|\mathbf{Z,C}\}}=\mu(\mathbf{Z},1,\mathbf{C})\tilde{S}.\]
Further,
\begin{align*}
\frac{E\{\mu(\mathbf{Z},S,\mathbf{C})\sigma^{-2}(\mathbf{Z},S,\mathbf{C})\tilde{S}|\mathbf{C}\}}{E\{\sigma^{-2}(\mathbf{Z},S,\mathbf{C})\tilde{S}^2|\mathbf{C}\}}=\frac{E\left\{\mu(\mathbf{Z},1,\mathbf{C})P_{\sigma}(\mathbf{Z,C})|\mathbf{C}\right\}}{E\left\{P_{\sigma}(\mathbf{Z,C})|\mathbf{C}\right\}}.
\end{align*}

Putting all this together, we have
\begin{align*}
&\left[\mu(\mathbf{Z},S,\mathbf{C})-\frac{E \{\mu(\mathbf{Z},S,\mathbf{C})\sigma^{-2}(\mathbf{Z},S,\mathbf{C})|\mathbf{Z,C}\}}{E\{ \sigma^{-2}(\mathbf{Z},S,\mathbf{C})|\mathbf{Z,C}\}}-\frac{E\{\mu(\mathbf{Z},S,\mathbf{C})\sigma^{-2}(\mathbf{Z},S,\mathbf{C})\tilde{S}|\mathbf{C}\}}{E\{ \sigma^{-2}(\mathbf{Z},S,\mathbf{C})\tilde{S}^2|\mathbf{C}\}}\tilde{S}\right]\\
&\times \sigma^{-2}(\mathbf{Z},S,\mathbf{C})\epsilon^*(\psi^\dagger)\\
&=\left[\mu(\mathbf{Z},1,\mathbf{C})-\frac{E\left\{\mu(\mathbf{Z},1,\mathbf{C})P_{\sigma}(\mathbf{Z,C})|\mathbf{C}\right\}}{E\left\{P_{\sigma}(\mathbf{Z,C})|\mathbf{C}\right\}}\right]\tilde{S}\sigma^{-2}(\mathbf{Z},S,\mathbf{C})\epsilon^*(\psi^\dagger)\\
&=\left[\mu(\mathbf{Z},1,\mathbf{C})-\frac{E\left\{\mu(\mathbf{Z},1,\mathbf{C})P_{\sigma}(\mathbf{Z,C})|\mathbf{C}\right\}}{E\left\{P_{\sigma}(\mathbf{Z,C})|\mathbf{C}\right\}}\right]P_{\sigma}(\mathbf{Z,C})\frac{(-1)^{1-S}}{f(S|\mathbf{Z},\mathbf{C})}\epsilon^*(\psi^\dagger)\\
&=\left[\mu(\mathbf{Z},1,\mathbf{C})-\frac{E\left\{\mu(\mathbf{Z},1,\mathbf{C})P_{\sigma}(\mathbf{Z,C})|\mathbf{C}\right\}}{E\left\{P_{\sigma}(\mathbf{Z,C})|\mathbf{C}\right\}}\right]P_{\sigma}(\mathbf{Z,C})\frac{(-1)^{1-S}f(\mathbf{Z}|\mathbf{C})}{f(S,\mathbf{Z}|\mathbf{C})}\epsilon^*(\psi^\dagger).
\end{align*}
This is of the form of equation \eqref{aid_rep_zs}, with admissible independence densities
\[
f^{\ddag }\left(\mathbf{Z},S|\mathbf{C}\right) =\left( \frac{1}{2}\right) ^{S}\left( \frac{1}{2}%
\right) ^{1-S}f(\mathbf{Z}|\mathbf{C})
\]%
and optimal index 
\[
r_{0}^{opt}(\mathbf{Z},S,\mathbf{C}) =\left\{\mu(\mathbf{Z},1,\mathbf{C})-\frac{E\left\{\mu(\mathbf{Z},1,\mathbf{C})P_{\sigma}(\mathbf{Z,C})|\mathbf{C}\right\}}{E\left\{P_{\sigma}(\mathbf{Z,C})|\mathbf{C}\right\}}\right\}2(-1)^{1-S}P_{\sigma }(\mathbf{Z,C}).
\]%
Note that        
\[
E^{\ddag }\left\{ r_{0}^{opt}(\mathbf{Z},S,\mathbf{C})|\mathbf{Z,C}\right\} =E^{\ddag }\left\{r_{0}^{opt}(\mathbf{Z},S,\mathbf{C})|S,\mathbf{C}\right\} =0
\]%
confirming that the optimal index function is in the correct subspace. 
Efficiency at the intersection submodel follows from general results in \citet{robins2001comment}.

If $\mathbf{Z}$ is binary, then one can obtain the efficient score by obtaining the optimal choice $m_{opt}(\mathbf{C})$ of $m(\mathbf{C})$ in 
\begin{align*}
\frac{m(\mathbf{C})(-1)^{Z+S}}{f(Z,S|\mathbf{C})}\epsilon^*(\psi^\dagger)
\end{align*}
which can be done via a population least squares projection of the score $U_{\psi^\dagger}$ for $\psi$ onto the above.
Thus,
\begin{align*}
 m_{opt}(\mathbf{C})&=E\left[\left\{\frac{(-1)^{Z+S}}{f(Z,S|\mathbf{C})}\right\}^2\sigma^{2}(Z,S,\mathbf{C})\bigg|C\right]^{-1}E\left\{U_{\psi^\dagger}\frac{(-1)^{Z+S}}{f(Z,S|\mathbf{C})}\epsilon^*(\psi^\dagger)\bigg|C\right\}\\
 &=E\left[\left\{\frac{(-1)^{Z+S}}{f(Z,S|\mathbf{C})}\right\}^2\sigma^{2}(Z,S,\mathbf{C})\bigg|C\right]^{-1}E\left\{E(U_{\psi^\dagger}\epsilon^*(\psi^\dagger)|Z,S,\mathbf{C})\frac{(-1)^{Z+S}}{f(Z,S|\mathbf{C})}\bigg|C\right\}\\
  &=E\left[\left\{\frac{(-1)^{Z+S}}{f(Z,S|\mathbf{C})}\right\}^2\sigma^{2}(Z,S,\mathbf{C})\bigg|C\right]^{-1}E\left\{\mu(Z,S,\mathbf{C})\frac{(-1)^{Z+S}}{f(Z,S|\mathbf{C})}\bigg|C\right\}\\
  &=E\left[\left\{\frac{(-1)^{Z+S}}{f(Z,S|\mathbf{C})}\right\}^2\sigma^{2}(Z,S,\mathbf{C})\bigg|C\right]^{-1}\{\mu(Z=1,S=1,\mathbf{C})-\mu(Z=0,S=1,\mathbf{C})\}.
\end{align*}
\end{proof}

\subsection{Proof of Theorem \ref{np_inf}}\label{sm:proof_np}

\begin{proof}
To obtain the efficient influence function for $\psi$,  we must find the random variable $G$ which satisfies:
\[\frac{\partial }{\partial r}\psi^*_r |_{r=0}=E\{G U(\mathbf{O};r)\}|_{r=0}\]
for $U(\mathbf{O};r)=\partial \log\{f(\mathbf{O};r)\}/\partial r$, where $f(\mathbf{O};r)$ is a one-dimensional regular parametric submodel of the nonparametric model for the observed data distribution.  By the product rule, 
\begin{align*}
\frac{\partial }{\partial r}\psi^*_r |_{r=0}=&\int \frac{\partial }{\partial r} \left(\frac{\delta^Y_r(\mathbf{c})-t_r(1,\mathbf{c})}{\delta^A_r(\mathbf{c})}\right)\bigg|_{r=0}dF(\mathbf{c}|A=1, S=1)\\
&+\int \beta(\mathbf{c})dF_r(\mathbf{c}|A=1, S=1)/\partial r|_{r=0}\\
&=(i)+(ii).
\end{align*}
If we first consider term $(ii)$:
\begin{align*}
(ii)&=E\left\{\beta(\mathbf{C})U(\mathbf{C}|A=1,S=1)|A=1,S=1\right\}\\
&=\frac{1}{f(A=1,S=1)}E\left\{AS\beta(\mathbf{C})U(\mathbf{C}|A,S)\right\}\\
&=\frac{1}{f(A=1,S=1)}E\left(AS\beta(\mathbf{C})[E\{U(\mathbf{O})|\mathbf{C},A,S\}-E\{U(\mathbf{O})|A,S\}]\right)\\
&=\frac{1}{f(A=1,S=1)}E\left[AS\left\{\beta(\mathbf{C})-\psi\right\}U(\mathbf{O})\right].
\end{align*}
Considering now term $(i)$,   then 
\begin{align*}
(i)&=\frac{1}{f(A=1,S=1)} E \left\{\frac{\partial }{\partial r}\frac{\delta_r^Y(\mathbf{C})-t_{r}(1,\mathbf{C})}{\delta^A_{r}(\mathbf{C})}|_{r=0}f(A=1,S=1|\mathbf{C})\right\}.
\end{align*}
Following the steps in the proof of Theorem 5 in  \citet{wang2018bounded}, 
\[\frac{\partial }{\partial r}E_r(Y|Z=z,S=s,\mathbf{C})|_{r=0}=E\left\{\frac{I(Z=z)I(S=s)}{f(Z,S|\mathbf{C})}\{Y-E(Y|Z,S,\mathbf{C})\}U(\mathbf{O})\bigg|\mathbf{C}\right\}\]
and therefore 
\[\frac{\partial }{\partial r}\{\delta_r^Y(\mathbf{C})-t_{r}(1,\mathbf{C})\}|_{r=0}=E\left\{\frac{(2Z-1)(2S-1)}{f(Z,S|\mathbf{C})}\{Y-E(Y|Z,S,\mathbf{C})\}U(\mathbf{O})\bigg|\mathbf{C}\right\}  \]
and 
\[\frac{\partial }{\partial r}\delta_r^A(\mathbf{C})|_{r=0}=E\left\{\frac{(2Z-1)S}{f(Z,S=1|\mathbf{C})}\{A-E(A|Z,S=1,\mathbf{C})\}U(\mathbf{O})\bigg|\mathbf{C}\right\}  \]
which gives us
\begin{align*}
(i)&=\frac{1}{f(A=1,S=1)} E \bigg[\frac{f(A=1,S=1|\mathbf{C})}{\delta^A(\mathbf{C})}\frac{(2Z-1)}{f(Z,S|\mathbf{C})} \\
&\quad\times \left\{(2S-1)\{Y-E(Y|Z,S,\mathbf{C})\}- S\beta(\mathbf{C})\{A-f(A=1|Z,S=1,\mathbf{C})\} \right\}U(\mathbf{O}) \bigg]\\
&=\frac{1}{f(A=1,S=1)} E \bigg\{\frac{f(A=1,S=1|\mathbf{C})}{\delta^A(\mathbf{C})}\frac{(2Z-1)}{f(Z,S|\mathbf{C})} \\
&\quad\times \bigg( S[Y-E(Y|Z,S=1,\mathbf{C})-\beta(\mathbf{C})\{A-f(A=1|Z,S=1,\mathbf{C})\}]\\&\quad-(1-S)\{Y-E(Y|Z,S=0,\mathbf{C})\} \bigg)U(\mathbf{O}) \bigg\}.
\end{align*}
We furthermore have that 
\begin{align*}
&Y-E(Y|Z,S=1,\mathbf{C})-\{A-E(A|Z,S=1,\mathbf{C})\}\beta(\mathbf{C})\\
&=Y-E(Y|Z=0,S=1,\mathbf{C})-\delta^Y(\mathbf{C})Z-\beta(\mathbf{C})A+\beta(\mathbf{C})f(A=1|Z=0,S=1,\mathbf{C})+\beta(\mathbf{C})\delta^A(\mathbf{C})Z\\
&=Y-E(Y|Z=0,S=1,\mathbf{C})-\beta(\mathbf{C})\delta^A(\mathbf{C})Z-t(1,\mathbf{C})Z\\&\quad-\beta(\mathbf{C})A+\beta(\mathbf{C})f(A=1|Z=0,S=1,\mathbf{C})+\beta(\mathbf{C})\delta^A(\mathbf{C})Z\\
&=Y-E(Y|Z=0,S=1,\mathbf{C})-\beta(\mathbf{C})\{A-f(A=1|Z=0,S=1,\mathbf{C})\}-t(1,\mathbf{C})Z
\end{align*}
such that 
\begin{align*}
(i)&=\frac{1}{f(A=1,S=1)} E \bigg\{\frac{f(A=1,S=1|\mathbf{C})}{\delta^A(\mathbf{C})}\frac{(2Z-1)}{f(Z,S|\mathbf{C})} \\
&\quad\times \bigg( S[Y-E(Y|Z=0,S=1,\mathbf{C})-\beta(\mathbf{C})\{A-f(A=1|Z=0,S=1,\mathbf{C})\}-t(1,\mathbf{C})Z]\\&\quad-(1-S)\{Y-t(1,\mathbf{C})Z-E(Y|Z=0,S=0,\mathbf{C})\} \bigg)U(\mathbf{O}) \bigg\}\\
&=\frac{1}{f(A=1,S=1)} E \bigg\{\frac{f(A=1,S=1|\mathbf{C})}{\delta^A(\mathbf{C})}\frac{(2Z-1)(2S-1)}{f(Z,S|\mathbf{C})} \\
&\quad\times [Y-E(Y|Z=0,S,\mathbf{C})-\beta(\mathbf{C})\{A-f(A=1|Z=0,S=1,\mathbf{C})\}S-t(1,\mathbf{C})Z]U(\mathbf{O}) \bigg\}.
\end{align*}
The main result then follows by combining $(i)$ and $(ii)$. 
\end{proof}

\section{Simulation studies}\label{sm:sims}

\subsection{Additional estimators}\label{sm:add_est}

\subsection*{Two-stage least squares}

Under NEM, a two-stage least squares estimator can be obtained by solving the equations 
\begin{align*}
0=\sum^n_{i=1}S_im(\textbf{Z}_{i},\mathbf{C}_i)\{Y_i-\beta(\mathbf{A}_i,\mathbf{C}_i;\psi^\dagger)S_i-t(\mathbf{Z}_i,\mathbf{C}_i;\nu^\dagger)-b_0(\mathbf{C}_i;\theta^\dagger_0)-b_1(\mathbf{C}_i;\theta^\dagger_1)S_i\}
\end{align*}
for $\psi^\dagger$ \citep{richardson2021bespoke}, where $m$ is of the same dimension of $\psi^\dagger$.  In practice, the unknown values $\nu^\dagger$ and $\theta^\dagger$ can be substitued by the estimates $\tilde{\nu}$ and $\hat{\theta}$ as described in Section \ref{sm:est_strat}. Since this approach relies on a parametric model for the outcome, one can set $E(S|\mathbf{Z,C})$ to zero in the equations for $\theta^\dagger$; this is what was done to obtain the estimator $\hat{\psi}_{TSLS}$ in the simulations.

\subsection*{G-estimation: two strategies}

One can construct a $g$-estimator via solving the equations
\begin{align*}
0=&\sum^n_{i=1}S_i[m(\textbf{Z}_{i},\mathbf{C}_i)-E\{m(\textbf{Z}_{i},\mathbf{C}_i)|S_i=1,\mathbf{C}_i;\tau^\dagger,\rho^\dagger\}\\&\times\{Y_i-\beta(\mathbf{A}_i,\mathbf{C}_i;\psi^\dagger)S_i-t(\mathbf{Z}_i,\mathbf{C}_i;\nu^\dagger)-b_0(\mathbf{C}_i;\theta^\dagger_0)-b_1(\mathbf{C}_i;\theta^\dagger_1)S_i\}
\end{align*}
for $\psi^\dagger$. \citet{richardson2021bespoke} note that the above equations are doubly robust, and yield an estimator that is unbiased under the union model $\mathcal{M}_{NEM}\cap(\mathcal{M}_1\cup\mathcal{M}_2)$ i.e. so long as either $f(\mathbf{Z}|S=1,\mathbf{C})$ or $b_0(\mathbf{C})$ and $b_1(\mathbf{C})$ are correctly modelled. In order to obtain a doubly robust estimator, $\nu^\dagger$ must be estimated using a estimator consistent under the laws $\mathcal{M}_1\cap \mathcal{M}_2$. In the simulations, $\nu^\dagger$ was estimated via the equations given in Section \ref{sec_mr} for $\hat{\nu}$, with $b_0(\mathbf{C})$ fixed at zero; similarly, in the $g$-estimating equations used to estimate $\psi$ as $\hat{\psi}_{g-Z}$, $b_0(\mathbf{C})$ and $b_1(\mathbf{C})$ were fixed at zero. 

An alternative doubly $g$-estimator can be obtained as $\tilde{\psi}$ based on the equations in Section \ref{sm:est_strat}, step 4; the resulting estimator is then unbiased under the union model $\mathcal{M}_{NEM}\cap(\mathcal{M}_1\cup\mathcal{M}_3)$. In order to construct the estimator $\hat{\psi}_{g-S}$ in the simulations, we set $t(\mathbf{Z,C})$ and $b_0(\mathbf{C})$ at zero.

\subsection*{Inverse probability weighted estimation}

An inverse probability weighted (IPW) estimator can be obtained via solving the estimating equations:
\[0=\sum^n_{i=1}\phi(\mathbf{Z}_i,S_i,\mathbf{C}_i;\tau^\dagger,\alpha^\dagger,\rho^\dagger)\{Y_i-\beta(\mathbf{A}_i,\mathbf{C}_i;\psi^\dagger)S_i\}\]
for $\psi^\dagger$. One can show along the lines of the proof of Theorem \ref{triply} that the resulting estimator is unbiased under the model $\mathcal{M}_{NEM}\cap\mathcal{M}_4$. In the simulations, we estimated $\tau^\dagger$, $\alpha^\dagger$ and $\rho^\dagger$ as proposed in Section \ref{sec_mr}, in order to construct the estimator $\hat{\psi}_{IPW}$

\subsection{Simulation set-up and results}

We conducted a series of simulation studies to first evaluate the robustness properties of the proposed estimators from Section \ref{sec_mr}. Specifically, we generated covariates $C_1$ and $C_2$ from a Bernoulli distribution with expectation 0.5 and a standard normal distribution respectively; let $\mathbf{C}=(C_1,C_2)$. An unmeasured confounder $U$ was also generated as a random binary variable mean with expectation 0.5. Then the sample selection probability was $f(S=1|\mathbf{C})=expit(-0.5+C_1+0.6C_2+0.5C_1C_2)$, and $Z$ was generated from a Bernoulli distribution where $f(Z=1|\mathbf{C})=expit(0.25C_1-0.25C_2+0.5C_1C_2)$. Generating $S$ in this manner led to around 60\% of individuals in a simulated data set being members of the population of interest ($S$=1). The exposure was simulated from a Bernoulli distribution with $f(A=1|Z,S=1,\mathbf{C},U)=expit(1-1.5Z-0.75C_1-0.3C_2-0.5C_1C_2+U)$ for individuals in the population $S=1$ and was fixed at zero for all others. Finally, we generated the outcome from the distribution 
\begin{align*}
\mathcal{N}(&1+U+0.5C_1+0.5C_2-0.5C_1C_2+Z(1-0.4C_1-0.4C_2+0.5C_1C_2)\\&+S(A+0.5C_1+0.5C_2+0.5C_1C_2),1).\end{align*} 
This data generating mechanism is compatible with Assumptions \ref{ref}-\ref{pem}, but not with conditional exchangeability if one does not have access to $U$. Nor would it be compatible with the instrumental variable conditions, since both the exclusion restriction and unconfoundedness are violated. 

In our simulations, we considered a suite of estimators of the conditional average treatment effect on the treated in the population of interest, discussed in the previous subsection. Nuisance parameters were estimated using maximum likelihood, and a sandwich estimator was obtained for each of the estimators given above. For comparison, we also considered a doubly robust $g$-estimator that is valid under a conditional exchangeability assumption \citep{robins1994correcting}, where a linear model was fitted for the outcome, and a logistic model for the exposure; both models were adjusted for $Z$, $C_1$, $C_2$ and a $C_1C_2$ interaction term. Furthermore, we also implemented a standard two-stage least squares estimator under the assumption that $Z$ is a valid instrument. For both stages we fit models that were linear in the instrument/exposure, $C_1$, $C_2$ and a $C_1C_2$ interaction term.

We considered five initial experiments:
\begin{enumerate}
    \item All nuisance models are correctly specified: the intersection model $\mathcal{M}_{NEM}\cap\mathcal{M}_1\cap \mathcal{M}_2 \cap  \mathcal{M}_3\cap  \mathcal{M}_4$ holds.
    \item $f(Z|S=0,\mathbf{C};\tau^\dagger)$ and $f(S|Z=0,\mathbf{C};\alpha^\dagger)$ were misspecified: $\mathcal{M}_{NEM}\cap \mathcal{M}_1$ holds.
    \item $f(S|Z=0,\mathbf{C};\alpha^\dagger)$, $b_0(\mathbf{C};\theta_0^\dagger)$ and $b_1(\mathbf{C};\theta_1^\dagger)$, were misspecified: $\mathcal{M}_{NEM}\cap \mathcal{M}_2$ holds.
   \item $f(Z|S=0,\mathbf{C};\tau^\dagger)$ ,   $b_0(\mathbf{C};\theta_0^\dagger)$ and $t(Z,\mathbf{C};\nu^\dagger)$, were misspecified: $\mathcal{M}_{NEM}\cap \mathcal{M}_3$ holds.
    \item $t(Z,\mathbf{C};\nu^\dagger)$,   $b_0(\mathbf{C};\theta_0^\dagger)$ and $b_1(\mathbf{C};\theta_1^\dagger)$, were misspecified: $\mathcal{M}_{NEM}\cap \mathcal{M}_4$ holds.
\end{enumerate}
Misspecification was induced via omission of interaction terms. For each setting, we simulated 2,000 data sets, with a sample size of 5,000. 

We also performed simulations to evaluate violations of the identification assumptions. In experiment 6, we changed the outcome generating mechanism to
\begin{align*}
\mathcal{N}(&1+U+0.5C_1+0.5C_2-0.5C_1C_2+0.5ZS(1-0.4C_1-0.4C_2+0.5C_1C_2)\\&+S(A+0.5C_1+0.5C_2+0.5C_1C_2),1).\end{align*} 
such that $Z$ has no association with $Y$ in the reference population, but does has an association in the population of interest. In experiment 7, we considered settings where the association between $Z$ and $A$ is weak; the exposure was now generated as $f(A=1|Z,S=1,\mathbf{C},U)=expit(1-0.25Z-0.75C_1-0.3C_2-0.5C_1C_2+U)$ for those with $S=1$. Since the aim was to isolate potential bias due to failure of the identification assumptions, all models in experiments 6 and 7 were correctly specfied.

\begin{table}
\caption{\label{sim_tab}
Simulations evaluating impact of model misspecification. Experiment (Exp), empirical bias for the conditional effect in the treated (Bias), empirical standard error (SE), coverage probability (Cov).}
\centering
\resizebox{\textwidth}{!}{
\begin{tabular}{rrrrrrrrrrrrrrrr}
  \hline
   & \multicolumn{3}{c}{{$\hat{\psi}_{TSLS}$}} & \multicolumn{3}{c}{{$\hat{\psi}_{g-Z}$}} & \multicolumn{3}{c}{{$\hat{\psi}_{g-S}$}} & \multicolumn{3}{c}{{$\hat{\psi}_{g-IPW}$}} & \multicolumn{3}{c}{{$\hat{\psi}_{MR-eff}$}}\\
 \cline{2-4}\cline{5-7}\cline{8-10}\cline{11-13}\cline{14-16}
 {Exp} & {Bias} & {SE} & {Cov} & {Bias} & {SE} & {Cov} & {Bias} & {SE} & {Cov} & {Bias} & {SE} & {Cov} & {Bias} & {SE} & {Cov}\\
 \hline
  1 & 0.01 & 0.26 & 0.95 & -0.00 & 0.41 & 0.95 & 0.01 & 0.23 & 0.96 & 0.01 & 0.33 & 0.95 & 0.01 & 0.23 & 0.95 \\ 
    2 & -0.00 & 0.27 & 0.95 & 0.08 & 0.39 & 0.94 & 0.18 & 0.26 & 0.89 & -0.54 & 0.28 & 0.48 & -0.01 & 0.24 & 0.94 \\ 
    3 & -0.59 & 0.25 & 0.35 & -0.02 & 0.41 & 0.95 & -0.30 & 0.22 & 0.73 & -0.36 & 0.28 & 0.74 & -0.01 & 0.23 & 0.95 \\ 
    4 & -0.51 & 0.26 & 0.48 & -1.13 & 0.36 & 0.10 & 0.00 & 0.24 & 0.95 & -0.25 & 0.32 & 0.85 & -0.00 & 0.23 & 0.95 \\ 
    5 & -0.58 & 0.25 & 0.34 & -0.42 & 0.36 & 0.80 & -0.14 & 0.22 & 0.92 & 0.00 & 0.33 & 0.95 & -0.00 & 0.23 & 0.95 \\ 
   \hline
\end{tabular}}
\end{table}

The results of the experiments 1-5 can be seen in Table \ref{sim_tab}. As a comparison, the benchmark doubly robust $g$-estimator has a bias in all experiments of 0.24, with coverage probability $<$0.01. Furthermore, the benchmark two-stage least squares estimator had a bias of -2.35 with coverage probability 0. We see that the bias and coverage properties of the estimators reflect the theory; the multiply robust estimator is the only approach that maintains good coverage properties across all experiments. The $g$-estimator $\hat{\psi}_{g-Z}$ exhibited wide confidence intervals across settings and sometimes large bias when the models were misspecified; improvements could potentially be made via modelling $b_0(\mathbf{C})$ and making the estimator doubly robust e.g. unbiased under model $\mathcal{M}_{NEM}\cap( \mathcal{M}_1\cup \mathcal{M}_2)$. The efficient multiply robust estimator outperformed the two-stage least squares estimator both in terms of bias in the presence of misspecification, but also precision when models were correct. The latter aspect is perhaps not totally surprising, given that two-stage least squares estimators are not necessarily efficient even when the conditional outcome mean is correctly specified \citep{vansteelandt2018improving}. Although we have only considered two covariates, we expect our estimators to scale with covariate dimension similarly to other M-estimators.

Results for experiments 6-7 can be seen in in Table \ref{sim_tab2}. The doubly robust $g$-estimator under no unmeasured confounding has a bias of 0.24 (coverage $<$0.01) in both experiments. The benchmark two-stage least squares estimator has a bias of -1.17 in experiment 6 and -16.8 in experiment 7 (coverage was zero in both experiments). We see that under violations of partial population exchangeability (experiment 6), estimators exhibit bias and all coverage guarantees are lost. Bias was nevertheless comparable to that of the benchmark two-stage least squares estimator (although this may be an artifact of the setting). When the association between $A$ and $Z$ was weakened, we see that estimators exhibit some finite sample bias, and potentially dramatic increases in variance (in particular, $\hat{\psi}_{g-IPW}$). Nevertheless the coverage of confidence intervals was conservative.

\begin{table}[htbp]
\caption{\label{sim_tab2}
Simulations evaluating impact of identification failure. Experiment (Exp), empirical bias for the conditional effect in the treated (Bias), empirical standard error (SE), coverage probability (Cov).}
\centering
\resizebox{\textwidth}{!}{
\begin{tabular}{rrrrrrrrrrrrrrrr}
  \hline
   & \multicolumn{3}{c}{$\hat{\psi}_{TSLS}$} & \multicolumn{3}{c}{$\hat{\psi}_{g-Z}$} & \multicolumn{3}{c}{$\hat{\psi}_{g-S}$} & \multicolumn{3}{c}{$\hat{\psi}_{g-IPW}$} & \multicolumn{3}{c}{$\hat{\psi}_{MR-eff}$}\\
 \cline{2-4}\cline{5-7}\cline{8-10}\cline{11-13}\cline{14-16}
 Exp & Bias & SE & Cov & Bias & SE & Cov & Bias & SE & Cov & Bias & SE & Cov & Bias & SE & Cov\\
 \hline
6 & -1.18 & 0.27 & 0.00 & -1.19 & 0.40 & 0.14 & -1.27 & 0.26 & 0.00 & -1.32 & 0.31 & 0.01 & -1.23 & 0.24 & 0.00 \\ 
  7 & 0.17 & 5.83 & 0.99 & 0.72 & 17.69 & 0.99 & 0.28 & 7.84 & 0.99 & -0.03 & 39.14 & 1.00 & 0.28 & 8.16 & 0.99 \\ 
   \hline
\end{tabular}}
\label{Kang2}
\end{table}

\bibliographystyle{apalike}
\bibliography{bibfile}

\begin{thebibliography}{}

\bibitem[Bickel et~al., 1993]{bickel1993efficient}
Bickel, P.~J., Klaassen, C.~A., Bickel, P.~J., Ritov, Y., Klaassen, J.,
  Wellner, J.~A., and Ritov, Y. (1993).
\newblock {\em Efficient and adaptive estimation for semiparametric models},
  volume~4.
\newblock Johns Hopkins University Press Baltimore.

\bibitem[Chen, 2007]{yun2007semiparametric}
Chen, H.~Y. (2007).
\newblock A semiparametric odds ratio model for measuring association.
\newblock {\em Biometrics}, 63(2):413--421.

\bibitem[Chernozhukov et~al., 2018]{chernozhukov2018double}
Chernozhukov, V., Chetverikov, D., Demirer, M., Duflo, E., Hansen, C., Newey,
  W., and Robins, J. (2018).
\newblock Double/debiased machine learning for treatment and structural
  parameters.
\newblock {\em The Econometrics Journal}, 21(1):C1--C68.

\bibitem[Dahabreh et~al., 2019]{dahabreh2019generalizing}
Dahabreh, I.~J., Robertson, S.~E., Tchetgen, E.~J., Stuart, E.~A., and
  Hern{\'a}n, M.~A. (2019).
\newblock Generalizing causal inferences from individuals in randomized trials
  to all trial-eligible individuals.
\newblock {\em Biometrics}, 75(2):685--694.

\bibitem[Danieli et~al., 2022]{danieli2022negative}
Danieli, O., Nevo, D., Walk, I., Weinstein, B., and Zeltzer, D. (2022).
\newblock Negative controls for instrumental variable designs.

\bibitem[Davies et~al., 2017]{davies2017compare}
Davies, N.~M., Thomas, K.~H., Taylor, A.~E., Taylor, G.~M., Martin, R.~M.,
  Munafo, M.~R., and Windmeijer, F. (2017).
\newblock How to compare instrumental variable and conventional regression
  analyses using negative controls and bias plots.
\newblock {\em International Journal of Epidemiology}, 46(6):2067--2077.

\bibitem[De~Chaisemartin and d’Haultfoeuille, 2018]{de2018fuzzy}
De~Chaisemartin, C. and d’Haultfoeuille, X. (2018).
\newblock Fuzzy differences-in-differences.
\newblock {\em The Review of Economic Studies}, 85(2):999--1028.

\bibitem[Hern{\'a}n and Robins, 2006]{hernan2006instruments}
Hern{\'a}n, M.~A. and Robins, J.~M. (2006).
\newblock Instruments for causal inference: an epidemiologist's dream?
\newblock {\em Epidemiology}, pages 360--372.

\bibitem[Kallus et~al., 2019]{kallus2019localized}
Kallus, N., Mao, X., and Uehara, M. (2019).
\newblock Localized debiased machine learning: Efficient inference on quantile
  treatment effects and beyond.
\newblock {\em arXiv preprint arXiv:1912.12945}.

\bibitem[Lipsitch et~al., 2010]{lipsitch2010negative}
Lipsitch, M., Tchetgen~Tchetgen, E.~J., and Cohen, T. (2010).
\newblock Negative controls: a tool for detecting confounding and bias in
  observational studies.
\newblock {\em Epidemiology (Cambridge, Mass.)}, 21(3):383.

\bibitem[Miao et~al., 2018]{miao2018identifying}
Miao, W., Geng, Z., and Tchetgen~Tchetgen, E.~J. (2018).
\newblock Identifying causal effects with proxy variables of an unmeasured
  confounder.
\newblock {\em Biometrika}, 105(4):987--993.

\bibitem[Richardson, 2012]{richardson2012lessons}
Richardson, D. (2012).
\newblock Lessons from hiroshima and nagasaki: The most exposed and most
  vulnerable.
\newblock {\em Bulletin of the Atomic Scientists}, 68(3):29--35.

\bibitem[Richardson and Tchetgen~Tchetgen, 2021]{richardson2021bespoke}
Richardson, D.~B. and Tchetgen~Tchetgen, E.~J. (2021).
\newblock Bespoke instruments: A new tool for addressing unmeasured
  confounders.
\newblock {\em American journal of epidemiology}.

\bibitem[Richardson et~al., 2023]{richardson2023generalized}
Richardson, D.~B., Ye, T., and Tchetgen, E. J.~T. (2023).
\newblock Generalized difference-in-differences.
\newblock {\em Epidemiology}, 34(2):167--174.

\bibitem[Robins, 1994]{robins1994correcting}
Robins, J.~M. (1994).
\newblock Correcting for non-compliance in randomized trials using structural
  nested mean models.
\newblock {\em Communications in Statistics-Theory and methods},
  23(8):2379--2412.

\bibitem[Robins et~al., 1992]{robins1992estimating}
Robins, J.~M., Mark, S.~D., and Newey, W.~K. (1992).
\newblock Estimating exposure effects by modelling the expectation of exposure
  conditional on confounders.
\newblock {\em Biometrics}, pages 479--495.

\bibitem[Robins and Rotnitzky, 2001]{robins2001comment}
Robins, J.~M. and Rotnitzky, A. (2001).
\newblock Comment on “inference for semiparametric models: Some questions and
  an answer,” by pj bickel and j. kwon.
\newblock {\em Statistica Sinica}, 11:920--936.

\bibitem[Robins et~al., 1994]{robins1994estimation}
Robins, J.~M., Rotnitzky, A., and Zhao, L.~P. (1994).
\newblock Estimation of regression coefficients when some regressors are not
  always observed.
\newblock {\em Journal of the American statistical Association},
  89(427):846--866.

\bibitem[Sanderson et~al., 2021]{sanderson2021use}
Sanderson, E., Richardson, T.~G., Hemani, G., and Davey~Smith, G. (2021).
\newblock The use of negative control outcomes in mendelian randomization to
  detect potential population stratification.
\newblock {\em International journal of epidemiology}, 50(4):1350--1361.

\bibitem[Sofer et~al., 2016]{sofer2016negative}
Sofer, T., Richardson, D.~B., Colicino, E., Schwartz, J., and
  Tchetgen~Tchetgen, E.~J. (2016).
\newblock On negative outcome control of unobserved confounding as a
  generalization of difference-in-differences.
\newblock {\em Statistical science: a review journal of the Institute of
  Mathematical Statistics}, 31(3):348.

\bibitem[Tchetgen~Tchetgen et~al., 2010]{tchetgen2010doubly}
Tchetgen~Tchetgen, E.~J., Robins, J.~M., and Rotnitzky, A. (2010).
\newblock On doubly robust estimation in a semiparametric odds ratio model.
\newblock {\em Biometrika}, 97(1):171--180.

\bibitem[Tchetgen~Tchetgen and Vansteelandt, 2013]{tchetgen2013alternative}
Tchetgen~Tchetgen, E.~J. and Vansteelandt, S. (2013).
\newblock Alternative identification and inference for the effect of treatment
  on the treated with an instrumental variable.

\bibitem[Tchetgen~Tchetgen et~al., 2024]{tchetgen2024introduction}
Tchetgen~Tchetgen, E.~J., Ying, A., Cui, Y., Shi, X., and Miao, W. (2024).
\newblock An introduction to proximal causal inference.
\newblock {\em Statistical Science}, 39(3):375--390.

\bibitem[Tsiatis, 2007]{tsiatis2007semiparametric}
Tsiatis, A. (2007).
\newblock {\em Semiparametric theory and missing data}.
\newblock Springer Science \& Business Media.

\bibitem[Vansteelandt and Didelez, 2018]{vansteelandt2018improving}
Vansteelandt, S. and Didelez, V. (2018).
\newblock Improving the robustness and efficiency of covariate-adjusted linear
  instrumental variable estimators.
\newblock {\em Scandinavian Journal of Statistics}, 45(4):941--961.

\bibitem[Vansteelandt et~al., 2008]{vansteelandt2008multiply}
Vansteelandt, S., VanderWeele, T.~J., Tchetgen~Tchetgen, E.~J., and Robins,
  J.~M. (2008).
\newblock Multiply robust inference for statistical interactions.
\newblock {\em Journal of the American Statistical Association},
  103(484):1693--1704.

\bibitem[Wang and Tchetgen~Tchetgen, 2018]{wang2018bounded}
Wang, L. and Tchetgen~Tchetgen, E. (2018).
\newblock Bounded, efficient and multiply robust estimation of average
  treatment effects using instrumental variables.
\newblock {\em Journal of the Royal Statistical Society: Series B (Statistical
  Methodology)}, 80(3):531--550.

\bibitem[Ye et~al., 2020]{ye2020instrumented}
Ye, T., Ertefaie, A., Flory, J., Hennessy, S., and Small, D.~S. (2020).
\newblock Instrumented difference-in-differences.
\newblock {\em arXiv preprint arXiv:2011.03593}.

\end{thebibliography}

\end{document}